\documentclass[conference]{IEEEtran}
\IEEEoverridecommandlockouts
\usepackage{cite}
\usepackage{amsmath,amssymb,amsfonts}
\usepackage{algorithmic}
\usepackage{multirow}
\usepackage{diagbox}
\usepackage{graphicx}
\usepackage{textcomp}
\usepackage{xcolor}
\usepackage{amsthm}
\usepackage{tabularx}
\usepackage{subfigure}
\usepackage{stfloats}
\usepackage{hyperref}
\usepackage[utf8]{inputenc}
\usepackage{babel,blindtext}
\usepackage{titlesec}
\usepackage{cuted}
\usepackage{siunitx,tabularx}

\newtheorem{definition}{Definition}
\newtheorem{proposition}{Proposition}
\newtheorem{corollary}{Corollary}
\newtheorem{property}{Property}
\newtheorem{theorem}{Theorem}

\usepackage{amssymb}
\usepackage{amsthm}
\usepackage{amsmath}
\usepackage{verbatim}

\def\BibTeX{{\rm B\kern-.05em{\sc i\kern-.025em b}\kern-.08em
    T\kern-.1667em\lower.7ex\hbox{E}\kern-.125emX}}
\begin{document}

\title{Performances of Symmetric Loss for \\ Private Data from Exponential Mechanism\\
}

\author{\IEEEauthorblockN{\textbf{Jing Bi, Vorapong Suppakitpaisarn} }
\IEEEauthorblockA{Graduate School of Information Science and Technology}
\text{The University of Tokyo}\\
}

\maketitle
\thispagestyle{plain}
\pagestyle{plain}

\begin{abstract}
This study explores the robustness of learning by symmetric loss on private data. Specifically, we leverage exponential mechanism (EM) on private labels. First, we theoretically re-discussed properties of EM when it is used for private learning with symmetric loss. Then, we propose numerical guidance of privacy budgets corresponding to different data scales and utility guarantees. Further, we conducted experiments on the CIFAR-10 dataset to present the traits of symmetric loss. Since EM is a more generic differential privacy (DP) technique, it being robust has the potential for it to be generalized, and to make other DP techniques more robust.
\end{abstract}


\section{Introduction}
Individual information privacy has become a severe issue with the development of data analysis techniques. Risk of exposing sensitive data would discourage respondents to cooperate honestly with our caring investigations. Even worse, it can also discourage the participation of respondents. Therefore, many researchers propose to corrupt a part of data before publication.
Data corruption and publication techniques that reveal the maximum possible amount of knowledge, with minimum possible risk of exposure of sensitive data, are imperative for data analysis. 

There are several notions proposed for private data publications such as $k$-anonymity~\cite{Sweene02} or $\ell$-diversity~\cite{machanavajjhala2007diversity}. One of the most commonly-known notions is DP \cite{Dwork06}. Many researchers are interested in DP because it can precisely quantify information leakage, usually called privacy budget, in each publication. 

Several researches in DP aim to find the most precise publication under a given privacy budget. Many, such as \cite{PETS22}, aim to publish a precise machine learning (ML) output under DP. However, none of them consider relationship between the precision and loss function, which is one of the most important components in several ML approaches.   


\subsection{Our Contributions}

In this study, we consider a type of loss function called symmetric loss. Symmetric loss, and, in particular, barrier hinge loss, are proven to be robust against corrupted data \cite{charoenphakdee2019symmetric}. As we corrupt data for privacy in DP, we believe that the loss is suitable for DP. Indeed, the corrupted model considered in \cite{charoenphakdee2019symmetric} is identical to the corruption in randomized response (RR) \cite{warner1965randomized, Yue2016}––which is one of the publication techniques under DP. We can immediately apply the results in \cite{charoenphakdee2019symmetric} to RR.

We question if symmetric losses can also be robust for other publication techniques. In this work, we consider publications under EM \cite{mcsherry2007mechanism, b1}. While RR requires that all data records must be independently sampled, we do not need that assumption in EM. Exponential mechanism is then more applicable than RR, and thus, is worth considering \footnote{Unfortunately, the robustness proof in \cite{charoenphakdee2019symmetric} assumes that the records are independently sampled. We are considering a robustness proof that does not require the independence.}. Our contributions can be summarized as follows: 
\begin{enumerate}
\item \textbf{Theoretical Results:} We presented and proved a series of properties between data scale and the expected utility range while implementing EM. Further, we provide numerical guidance of privacy budgets for EM. In particular, we provide theoretical evidence that an accurate learning output can be obtained with almost certain probability using EM and symmetric loss.
\item \textbf{Experimental Results:} We investigated the robustness of barrier hinge loss on an EM processed CIFAR-10 dataset. The experiment results have been presented. Even with a small privacy budget $\epsilon =0.5$, the output accuracy could reach a level of $80\%$. The accuracy is larger for larger $\epsilon$. In this sense, we are allowed to use highly protected data. 
\item We corrected the mistake in the formula for implementing EM in \cite{b1}.
\end{enumerate}

\subsection{Related Works}
Current DP research includes two major directions. The first direction is concerned with how to quantify privacy distances. The original definition of DP is sometimes too trivial for different privacy demands. Therefore, providing a wide range of definitions, such as relaxed $(\epsilon,\delta) $-DP \cite{Dwork2014}, R\'{e}nyi-DP \cite{Mironov2017}, and so forth––provides higher dimensions for researchers to demonstrate inspirations for new mechanisms. 

The second direction is to provide new mechanisms, several of which are for ML applications. As ML and DP are both data based, applying ML approaches in DP mechanisms has become popular. Specifically, as label-DP regime \cite{PETS22} naturally fits the ML methods, rolling-up ML approaches for better utility DP frameworks is simple, such as in \cite{Abadi2016, PETS22}. In \cite{PETS22}, many aggregated ML-DP frameworks are presented.

Our discussion is different from the above two directions as follows. Firstly, this work is still based on the traditional DP definition. Therefore, this research can easily be extended to other privacy quantifications. Further, we do not focus on aggregating multiple methods to make new DP frameworks. We only present performances of loss functions on EM processed data. Unlike the examples in \cite{PETS22} that purely pursue better performance with unsupervised or ensemble learning pre-processing, our experiments still demonstrate the insufficiency of only depending on loss functions––especially while the privacy budget is small. However, as loss functions are fundamental elements in ML, testing them can provide a reference of modifications for a series of ML-DP hybrid frameworks.
\subsection{Paper Organization}

The rest of the paper is organized as follows. In section II, we present an introduction to symmetric loss with its properties, and EM with its implementation. In section III, we discuss some theoretical properties of EM. In section IV, we show the experiment results on the CIFAR-10 dataset, which shows the robustness of barrier hinge loss on EM. 

According to PDAA regular paper format, proofs of properties in section III are moved to \ref{app1}ppendix, which is not included in the publication. One can find the full version on the ArXiv.

\section{Preliminaries}

\subsection{Symmetric Loss in Machine Learning}
Any given data can be represented by a series of real values, a $(d+k)$-tuple $(\mathbf{x} \in \mathbb{R}^d,\mathbf{y} \in \mathbb{R}^k)$. In ML, $\mathbf{x}$ are called features, and $\mathbf{y}$ are called labels. Our dataset $D$ is a collection of those tuples; i.e., $D := \langle \mathbf{x}_i, \mathbf{y}_i\rangle_{i = 1}^n$ is a dataset with size $n$. It is not necessary to have all meaningful values for all entries in a tuple in a dataset $D$. For example, we may have $0$'s or wrong values for some entries. Normally, features are treated as fully and correctly known elements; otherwise, we can move the unclear entry or entries to the label side. 

Machine learning is meant to find the internal relationship between features and labels in a set, and make predictions for prospective data. Traditionally, if all $\mathbf{y}$'s are correctly valued, the process is called supervised learning; if they are partially unclear, the process is called semi-supervised learning; if all labels are unknown, the process is called unsupervised learning. Formally, a function $f:\mathbb{R}^d \rightarrow \mathbb{R}^k$, which can describe this relation as best as possible for the whole dataset is the learning goal (i.e., $|f(\mathbf{x})-\mathbf{y}|=0$ for all $(\mathbf{x},\mathbf{y})$'s is desired). 

In this work, we consider binary classification problems, which is the task when $\mathbf{y} \in \{1,-1\}$. We call $(\mathbf{x},+1)$ and $(\mathbf{x},-1)$ as positive and negative data, respectively. We also define $D_P = \{(\mathbf{x},\mathbf{y}) \in D: \mathbf{y} = 1\}$ and $D_N = \{(\mathbf{x},\mathbf{y}) \in D: \mathbf{y} = -1\}$.

The function $f$ is controlled by parameters and hyper-parameters. The learning process is to optimize these parameters through given objectives. An intuitive objective, for example, can be Bayes risk $R(f)=\mathbb{E}[l(f(\mathbf{x}),\mathbf{y})]$, where $l$ satisfying the definition of the distance function is called the loss function. Since it is impossible to know and calculate all possible or potential data in the same problem, we normally restrict the expectations to certain sizes, which are named empirical risks, in the form of $\frac{1}{|D|}\sum_{(\mathbf{x},\mathbf{y}) \in D} l(f(\mathbf{x}),\mathbf{y})$. 

Trivial empirical risk unbiasedly weighs the losses from every data point, which outputs functions that suffer from disparate data treats. Manipulating it is necessary and, and different predictors can clearly be derived. In binary classification problems, to fairly evaluate the errors from positive and negative labeled data, one would leverage the receiver operating characteristic curve (AUC) \cite{Narasimhan2016}, or the balanced error rate (BER) \cite{Feldman2014, Zhao2020}. AUC and BER are defined in the following definitions.

\begin{definition}
Empirical AUC-risk is defined as: 
$$R_{AUC}=\mathbb{E}_{D_P}[\mathbb{E}_{D_N}[l(f(\mathbf{x}_P)-f(\mathbf{x}_N)]]$$
$$=\frac{\sum_{(\mathbf{x}_P,1) \in D_P}\sum_{(\mathbf{x}_N,-1) \in D_N} l(f(\mathbf{x}_P)-f(\mathbf{x}_N))}{|D_P||D_N|}$$
\end{definition}
\begin{definition}
Empirical BER-risk is defined as: 
$$R_{BER}=\frac{1}{2}\left[\mathbb{E}_{D_P}[l(\mathbf{y}f(\mathbf{x}))]+ \mathbb{E}_{D_N}[l(\mathbf{y}f(\mathbf{x}))]\right]$$
$$=\frac{1}{2}\left[\sum_{(\mathbf{x},1) \in D_P}l(f(\mathbf{x})) + \sum_{(\mathbf{x},-1) \in D_N}l(-f(\mathbf{x})) \right]$$
\end{definition}
The restrictions to certain datasets still carry in distances between empirical minimizers and Bayes minimizers. Let $f^*$ be the optimal function that describes the relation of the problem, and $\mathcal{F}$ is a collection of functions restricted to $D$. Such distances are described in a Bayes sense $R(f^*)- \inf_{f\in \mathcal{F}}R(f)$, named excess risks \cite{Bartlett2006}. Classical excess risks are defined on a 0-1 loss function, which is the function $l_{01}$ such that $l_{01}(x) = 1$ for $y \leq 0$, and $l_{01}(x) = 0$ otherwise. However, because of the non-differentiability and non-convexity of a 0-1 loss at the crucial point $0$, which harms the learning fluency and cost, we normally adopt other (surrogate) losses instead––like a logic sigmoid function, or hinge functions, and so forth. 

The unavoidable issue is how to guarantee the minimizer of surrogate excess risk to coincide with the minimizer of 0-1 excess risk. This fundamental property is called \textit{classification-calibration}, and has been proven to be satisfied for most commonly used surrogate losses (e.g., exponential, quadratic, hinge, and sigmoid \cite{Bartlett2006}). \cite{charoenphakdee2019symmetric} has extended this property to also be satisfied by symmetric functions, which are function $f$'s such that there is $C$ where $f(x) + f(-x) = C$ for all $x$. (Therefore, the function is symmetric about the intersection point with $y$-axis, not symmetric about $y$-axis.) Despite this merit, normal symmetric functions cannot simultaneously satisfy non-negative values, and convexity \cite{Plessis2014, Ghosh2014} (except for the horizontal lines). To cope with the issue, \cite{charoenphakdee2019symmetric} proposes the use of a restricted symmetric loss, called a barrier hinge function, which can be defined as follows:
\begin{definition}[\cite{charoenphakdee2019symmetric}] A barrier hinge function with parameter $b,r$ can be defined as follows:
$$f_{b,r}(x)=\max(-b(r+x)+r,\max(b(x-r),r-x)).$$
\end{definition}
The barrier hinge function is symmetric on $x\in [-r,r]$, when $b > 1$ and $r > 0$. 
Therefore, it satisfies \textit{classification-calibration} for the sensitive interval $[-r,r]$, and satisfies the non-negativity and the convexity conditions at the same time.

In our experiments, we leverage both AUC and BER risks as complements to each other to evaluate the overall performances of several surrogate losses. Further, we also show the superiority of the barrier hinge loss over other losses on mild privacy cost.

\subsection{Exponential Mechanism}
Differential privacy \cite{Dwork06} aims to protect sensitive information of each individual with minimum harm to the utilities (learning results). The strategy of adding noise is to obfuscate every possible dataset $D$ with its most likely set of datasets $\{D'|D'\sim D\}$ with a distribution corresponding to the similarity (e.g., Hamming distance as in the following). Formally, we say a method or a randomized algorithm $\mathcal{A}$ is $\epsilon$-differential privacy preservative \cite{Dwork06} if, for all possible measurements $\mathcal{M}$, $$\frac{\Pr\left[\mathcal{A}(D)\in \mathcal{M}\right]} { \Pr\left[\mathcal{A}(D')\in \mathcal{M}\right]} \in [\exp(-\epsilon), \exp(\epsilon)],$$ for any given privacy budget $\epsilon$.

Label DP \cite{PETS22} is a relaxation version of DP, which constrains the private values to only lie on a subset of columns in an information table. Since the result of this relaxation fits the ML regime naturally, without any extra work, ML models could be interspersed. \cite{PETS22} demonstrated their frames of leveraging unsupervised and semi-supervised learning models in improving utilities of private data. 

One of most classical and wildly used DP methods is RR, introduced in \cite{warner1965randomized}. By simply randomly flipping the data with a fixed probability, we are able to protect individual privacy, and infer the classification ratio of the sensitive information. However, this classical method requires the plausibly deniable data to be mutually independent. Otherwise, flipping a part of the dependent data would not be sufficient. Unfortunately, several existing datasets are not entirely independent––thus, RR would not be efficient because of this. Instead, researches prefer to adopt EM, an upgraded version from Laplace mechanism \cite{Dwork2014}. 
\begin{definition}[Exponential Mechanism (EM) \cite{mcsherry2007mechanism}] Given any real valued function $q:\mathcal{D} \rightarrow \mathbb{R}$, with $\Delta q := \max_{D\sim D'}|q(D)-q(D')|$. The EM can pose any possible dataset $D'' \in \mathcal{D}$ with the probability $P(D'') \propto \exp(\frac{q(D'',D^*) \epsilon}{2\Delta q})$.
\end{definition}
\begin{proposition}
Exponential mechanism is $\epsilon$-differentially private.
\end{proposition}

Defining the real valued function $q$ is the essential step \cite{mcsherry2007mechanism}, which provides the original power of this mechanism. The global sensitivity $\Delta q$ of the maximum distance of adjacent datasets is used for binding the range of the distribution.

One common approach in defining $q$ is based on Hamming distance \cite{Jordan2019}: Let $\delta(D = \langle \mathbf{x}_i, y_i \rangle_{i=1}^n,D' = \langle \mathbf{x}_i, y'_i \rangle_{i=1}^n):= |\{i| y_i \not= y'_i\}|$. We define the Hamming distance $q$ as $q(D,D^*) :=|D_P\cap D_P^*|+|\overline{D_P\cup D_P^*}| =|D^*|- \delta(D,D^*)$~\cite{b1}.
The implementation of EM contains the calculations of the probability distribution over $\mathcal{D} := \{D' = \langle \mathbf{x}_i, y'_i \rangle_{i=1}^n : \langle y_i \rangle_{i = 1}^n \in \{-1,+1\}^n\}$. A straightforward implementation would cost $O(n^2)$ time and space. Redundant exponential function calculations carry an extra burden. Therefore, we apply the two-step approach \cite{b1}. Firstly, we compute the distribution for each possible score by logarithmic recursion. Please note that there is a misprint in \cite{b1} where they drop the term $\epsilon/2$ there.
\begin{equation*}
\log \Pr(q = i) 
= \begin{cases}
        -|D^*|\log(1+\exp(\frac{\epsilon}{2\Delta})) & {\text{when }i =0}\\
        \log(|D^*|-i+1)-\log i\\~~+\frac{\epsilon}{2}+\log \Pr(q = i-1) &  {\text{when }i > 0.}
\end{cases}
\end{equation*}
Then, we uniform-randomly select one table within the same score group, i.e., from the selected value of $q$, and select $D$ from the set $\{D' \in \mathcal{D} : q(D', D^*) = q\}$ with uniform probability. This two-step method cuts down the time and space complexity to a linear size. One noticeable point is, by this definition of $q$, EM will degrade to RR when the data size goes to infinity. Despite this odd fact, this degradation would not happen in real situations, but is still worth our study. Therefore, it is an open problem to theoretically prove the relation between EM and RR.

\section{Properties in Exponential Mechanism}

There is however one issue of the EM for binary classification using semi-supervised learning. The table, which we use for the classification, might be heavily flipped such that the learning algorithms cannot correctly learn the function. When we use symmetric loss, it is known that a learning algorithm can still provide a good result if the flipping rate is no more than half of the dataset––i.e., $q \geq |D|/2$ \cite{charoenphakdee2019symmetric}. In this section, we consider the probability that $q \geq |D|/2$ for each dataset size $n = |D|$, and each privacy budget $\epsilon$. Specifically, we consider $\sum_{i \geq |D|/2 }\Pr[q = i]$. We conduct theoretical analyses, then give numerical guidance for real implementations. 

In the following subsections, we only give out the short and straight proofs for corollaries. The important proofs of properties are demonstrated in the appendices. 

\subsection{Success probability and the dataset size}

In this subsection, we study the relationship between $\sum_{i \geq n/2 }\Pr[q = i]$, and the size of dataset $n$. We begin our illustration by proving generic properties of a binomial distribution.

\begin{definition}
A truncated binomial distribution $S$ with $n$ trials is the partial sum of the binomial distribution restricted to some certain number of events. For example, between $j$ and $k$: 
$$S(n,j,k)=\Pr\left[j\leq I_{n,p} \leq k\right]=
\sum_{i=j}^{k} \binom{n}{i}p^i(1-p)^{n-i},$$ where $p$ is the probability of the binomial distribution.
\end{definition}
\begin{definition} The upper truncated binomial distribution by 
$$S(n,j)=\Pr\left[ I_{n,p} \leq j\right]=\sum_{i=0}^{j}\binom{n}{i} p^i(1-p)^{n-i}.$$
In the lower truncated case, we may just use $1-S(n,j)$.
\end{definition}


\begin{property} Truncated binomial distribution $S(n,j)$ is monotonically decreasing w.r.t. $n$, while $j\in [0,n]$ is fixed
\end{property}

\begin{property}\label{prop2} Truncated binomial distribution $S(n,n-k)$ is monotonically increasing w.r.t. $n$, while $k \in [0,n_0]$ and $n_0$ are fixed.
\end{property}
These two properties seem to be trivial, however, the proofs are not intuitive, see \ref{app1}ppendix. Moreover, above properties are simple enough to not involve probability $p$. With the analysis getting deeper, we are able to deduce the following property considering $p$ with the help of above properties.

\begin{figure}[htbp]
    \centering
    \includegraphics[scale=0.56]{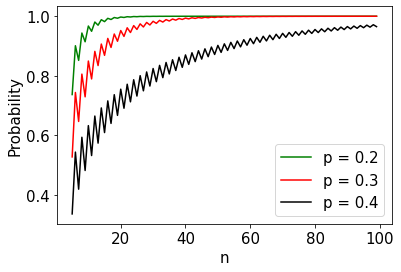}
    \caption{Trends of truncated binomial distributions $S(n,\lceil\frac{n}{2} \rceil)$ for $n \in [5,100]$ w.r.t. $p=0.2$, $p=0.3$ and $p=0.4$. All monotonicities except non-absolute ones have been shown.}
    \label{fig:my_label}
\end{figure}

\begin{property}
\label{prop3}
Truncated binomial distribution $S(n,\lceil\frac{n}{2} \rceil+k)$ is almost monotonic w.r.t. $n$, except for when $n$ changes between even and odd––this is while $k \in [-\lceil\frac{n_0}{2} \rceil, \lfloor \frac{n_0}{2} \rfloor]$, and $n_0$ is fixed. Moreover, if $p \leq \frac{n-\lceil\frac{n}{2} \rceil-k}{n+1} $, $S$ is almost monotonically increasing; and is decreasing if $p \geq \frac{n-\lceil\frac{n}{2} \rceil-k}{n+1} $, respectively. 
\end{property} 
The relationship between $n$ and $S(n,j)$ is illustrated as discussed in Property \ref{prop3} in Fig. \ref{fig:my_label}. This observation actually counters the normal intuition, we demonstrate it out for appliers not falling into intuition traps. We then use the property in the following corollary.

\begin{corollary}
The probability that the score $i$ taken in EM, $\sum_{i \geq |D|/2} \Pr[q = i]$, is no less than half of the maximum possible is almost monotonically increasing w.r.t. the data size $n$.
\end{corollary}
\begin{proof}
Firstly, such a probability is given by: 
\begin{eqnarray*}
&&\sum_{i \geq |D|/2} \Pr[q = i]\\ 
&=&  \sum_{i=\lceil \frac{n}{2}\rceil}^{n}\frac{{n\choose i} \exp(\frac{i\epsilon}{2\Delta})}{(1+\exp(\frac{\epsilon}{2\Delta}))^n}\\
&=&\sum_{i=\lceil \frac{n}{2}\rceil}^{n} {n\choose i} \left(\frac{1}{1+\exp(\frac{\epsilon}{2\Delta})}\right)^{n-i} \left(\frac{\exp(\frac{\epsilon}{2\Delta})}{1+\exp(\frac{\epsilon}{2\Delta})}\right)^i\\
&=& S(n,\lfloor \frac{n}{2}\rfloor)\\
&=&  S(n,\lceil \frac{n}{2}\rceil -1),
\end{eqnarray*}
which is exactly a truncated binomial distribution with the probability $p=\frac{1}{1+exp(\frac{\epsilon}{2\Delta})}$. As $\epsilon \geq 0$ and $\Delta >0$, then $$p \leq \frac{1}{2} \leq \frac{\lfloor\frac{n}{2} \rfloor+1}{n+1} =\frac{n-\lceil\frac{n}{2} \rceil+1}{n+1}.$$ Then, by Property \ref{prop3}, our probability $\sum_{i \geq |D|/2} \Pr[q = i]$ is almost monotonically increasing w.r.t. $n$.
\end{proof}
\subsection{Success probability and privacy budget}

Now, let us consider the relation between the probability $\sum_{i \geq n/2} \Pr\left[q = i\right]$ and the privacy budget $\epsilon$.

\begin{property}
The truncated binomial distribution $\Pr\left[I_{n,p} \leq j\right]$ is monotonically decreasing w.r.t. probability $p$. Moreover, if $p=0$, then $\Pr\left[I_{n,p} \leq j\right]=1$ constantly; if $p=1$, then $\Pr\left[I_{n,p} \leq j\right]=0$.
\end{property}
\begin{corollary}
The probability that the score taken in EM is greater than or equal to half of the possible maximum, $\sum_{i \geq n/2}\Pr[q = i]$, is monotonically increasing w.r.t. the privacy budget $\epsilon$. Further, $\lim\limits_{\epsilon \rightarrow \infty}\sum\limits_{i \geq n/2}\Pr[q = i] =1$. 
\end{corollary}
\begin{proof}
Since the mean of the flipping probability $p=\frac{1}{1+\exp(\frac{\epsilon}{2\Delta})}$ is negatively related to $\epsilon$, then it is trivial.
\end{proof}

\subsection{Recommendations on privacy budget decisions}

By the results in the previous two subsections, we can recommend the appropriate value of $\epsilon$ on each $n$. First, we recall that the results we have so far are as follows:
\begin{theorem} The probability that the score taken in EM is greater than or equal to half of the possible maximum, $\sum_{i \geq n/2}\Pr[q = i]$, is (almost) monotonic w.r.t. $n$ and $\epsilon$ when the limits tend to $1$.
\end{theorem}
Above results are clean from involving truncated point $j$, because our proofs did not involve the \textit{concentration inequalities} \cite{HighDim2018}. To apply those inequalities, we need to discuss the sign of the subtraction $j-E(\sum_{i=1}^{n}X_i) =j- np$. Therefore, we are forced to discuss the relation between $p$ with the truncation parameter $j$, which would extend to unnecessarily complicated proof procedures for above clean results. Adding more factors to our analysis, we will use \textit{central limit theorem} and \textit{concentration inequalities} for the following proposition, see \ref{app1}ppendix. And as one can see, $j$ is also involved in the relations.

\begin{table}[htbp]
\caption{referred to privacy budgets for 99.9\% confidence \\(empty cells are not applicable)}
\centering
\setlength{\tabcolsep}{0.6mm}{
\begin{tabular}{|c|c|c|c|c|c|c|c|c|c|c|}
\hline
\diagbox[innerwidth=1cm]{$n$}{Flip} & 50\% & 45\% & 40\% & 35\% & 30\% & 25\% & 20\% & 15\% & 10\% & 5\% \\
\hline
100 & 1.562 & 2.049 & 2.600 & 3.255 & 4.098 & 5.360 & 8.487 &  &  &  \\
\hline
1000 & 0.472 & 0.884 & 1.316 & 1.779 & 2.292 & 2.884 & 3.610 & 4.597 & 6.293 &  \\
\hline
10000 & 0.149 & 0.552 & 0.967 & 1.404 & 1.875 & 2.401 & 3.014 & 3.777 & 4.847 & 6.857 \\
\hline
100000 & 0.047 & 0.449 & 0.860 & 1.290 & 1.751 & 2.260 & 2.847 & 3.563 & 4.529 & 6.151 \\
\hline
1000000 & 0.015 & 0.416 & 0.826 & 1.254 & 1.712 & 2.217 & 2.796 & 3.499 & 4.436 & 5.969 \\
\hline
\end{tabular}}
\end{table}
\begin{table}[htbp]
\caption{referred to privacy budgets for 95\% confidence \\(empty cells are not applicable)}
\centering
\setlength{\tabcolsep}{0.6mm}{
\begin{tabular}{|c|c|c|c|c|c|c|c|c|c|c|}
\hline
\diagbox[innerwidth=1cm]{$n$}{Flip} & 50\% & 45\% & 40\% & 35\% & 30\% & 25\% & 20\% & 15\% & 10\% & 5\% \\
\hline
100 & 0.999 & 1.438 & 1.913 & 2.444 & 3.065 & 3.844 & 4.950 & 7.123 &  &  \\
\hline
1000 & 0.310 & 0.717 & 1.139 & 1.588 & 2.078 & 2.634 & 3.297 & 4.155 & 5.458 & 8.944 \\
\hline
10000 & 0.098 & 0.501 & 0.913 & 1.347 & 1.813 & 2.330 & 2.929 & 3.668 & 4.683 & 6.476 \\
\hline
100000 & 0.031 & 0.433 & 0.843 & 1.272 & 1.732 & 2.239 & 2.821 & 3.531 & 4.482 & 6.058 \\
\hline
1000000 & 0.010 & 0.411 & 0.821 & 1.249 & 1.706 & 2.210 & 2.788 & 3.488 & 4.422 & 5.941 \\
\hline
\end{tabular}}
\end{table}
\begin{proposition} For a given dataset of size $n$, to leverage EM, with minimum probability $P$, the score should be no less than $n-j$, i.e., $\sum_{i\geq n - j} \Pr[q=i] \geq P$, and we need a privacy budget to be no less than $2\Delta \log\frac{n-j+\sqrt{-n\log(1-P)/2}}{j-\sqrt{-n\log(1-P)/2}}$.
\label{pros2}
\end{proposition}
\begin{proof}
Set $p=\frac{1}{1+\exp(\frac{\epsilon}{2\Delta})} < \frac{1}{2}$.
We use $X_i$ to denote the event for each trial. We then have $E[\sum_{i=1}^{n}X_i] = np$, and we apply Hoeffding's inequality \cite{HighDim2018} of $0\leq X_i \leq 1$ for all $i$:
\begin{eqnarray*}
&&\sum_{i\geq n - j}  \Pr[q=i]\\
&=& \Pr\left[I_{n,p} \leq j\right]\\
&=& \Pr\left[\sum_{i=1}^{n}X_i \leq j\right] = 1 -\Pr\left[\sum_{i=1}^{n}X_i > j\right] \\
&=& 1 -\Pr\left[\sum_{i=1}^{n}X_i -np> j-np\right] \\
&\geq& 1 - \exp\left(-\frac{2(j-np))^2}{n}\right)\\
&\geq& P 
\end{eqnarray*}
implies
$$\exp\left(-\frac{2(j-np))^2}{n}\right) \leq 1- P,$$
$$p \leq \frac{j-\sqrt{-n\log(1-P)/2}}{n}.$$
By substituting back $\epsilon$, we obtain
$$\epsilon \geq 2\Delta \log\frac{n-j+\sqrt{-n\log(1-P)/2}}{j-\sqrt{-n\log(1-P)/2}}.$$

\end{proof}
\begin{corollary}
To ensure a $99.9\%$ probability of having a flipping rate that is no greater than $\frac{1}{2}$, we need a privacy budget to be no less than $2\Delta \log\frac{1+2\sqrt{(1.5\log10)/n}}{1-2\sqrt{(1.5\log10)/n}}$.
\end{corollary}

The tables on the left side provide the suggestions for privacy budgets of commonly referred to scenarios with sensitivity $1$.\\\\
\noindent\textbf{Remark: }The cases for the speeds of the increments of truncated points to be $\frac{1}{2}$ (flipping rate no more than $\frac{1}{2}$) of $n$ have been proven. We may further derive the results for different increment speeds following similar steps. However, since other speed factors are not our major concerns in our experiments of EM, we will not spend space to discuss here. 
\subsection{Exponential mechanism for large dataset}
Our analyses on the success probability has been completed in the previous subsection. In this section, we use the results obtained to show a relationship between the EM and the randomized response in large dataset. 

After the trend of truncated binomial distribution has been learned, we may further investigate the behaviors in extreme situations––for example, $n\rightarrow \infty$. According to De Moivre–Laplace theorem (Central Limit Theorem for binary distributions \cite{Ajoy2018}), a binomial distribution will converge to a normal distribution $\mathcal{N}(np, np(1-p))$. Moreover, the probability mass will converge to the probability density:
$${n\choose i}p^i(1-p)^{n-i} \overset{a}{\rightarrow} \frac{1}{\sqrt{2\pi n p(1-p)}}\exp\left(-\frac{(i-np)^2}{2np(1-p)}\right).$$
Hence, we may use an integral to approximate our truncated sum:
$$\sum_{i=j}^{k} {n\choose i}p^i(1-p)^{n-i}$$
$$\overset{a}{\rightarrow} \frac{1}{\sqrt{2\pi n p(1-p)}} \int_{j}^{k} \exp\left(-\frac{(i-np)^2}{2np(1-p)}\right) di.$$

Now, let us omit factor $k$ for context simplicity. 

\begin{property}\label{prop5}
The truncated binomial distribution $S(n,\lceil\frac{n}{2} \rceil)$ approaches $1$ if $p <\frac{1}{2}$, $0$ if $p >\frac{1}{2}$, and $\frac{1}{2}$ if $p=\frac{1}{2}$, as $n \rightarrow \infty $.
\end{property}
\begin{corollary} For any $p < \frac{1}{2}  $ (or $p > \frac{1}{2} $ respectively), there always exists a positive integer $R$ such that $S(n+r,\lceil\frac{n+r}{2} \rceil+k) \geq S(n,\lceil\frac{n}{2} \rceil+k)$ when for all $r > R$ (or $S(n+r,\lceil\frac{n+r}{2} \rceil+k) \leq S(n,\lceil\frac{n}{2} \rceil+k)$, respectively). 
\end{corollary}
\begin{proof}
By Property \ref{prop3}, we know $ S(n,\lceil\frac{n}{2} \rceil+k)$ is monotonic w.r.t. parities. However, by Property \ref{prop2} $$S(2m+1,\lceil\frac{2m+1}{2} \rceil+k) \leq S(2m,\lceil\frac{2m}{2} \rceil+k).$$ Moreover, it is not necessary that for a fixed odd number $R_0$ is such that $$S(2m+R_0,\lceil\frac{2m+R_0}{2} \rceil+k) \geq S(2m,\lceil\frac{2m}{2} \rceil+k).$$
And, Property \ref{prop5} tells us that
$$\lim\limits_{odd\ R_0 \rightarrow \infty}S(2m+R_0,\lceil\frac{2m+R_0}{2} \rceil+k) =1 \geq S(2m,\lceil\frac{2m}{2} \rceil+k).$$
Thus, using the definition of the limit, some interchange point $R$ must exist.
\end{proof}

\begin{property} For any $\delta \in (0,1]$, the truncated binomial distribution $S(n,n(p-\delta),n(p+\delta))$ approaches $1$ as $n \rightarrow \infty $.
\label{prop6}
\end{property}
\begin{figure}[htbp]
    \centering
    \includegraphics[scale=0.2]{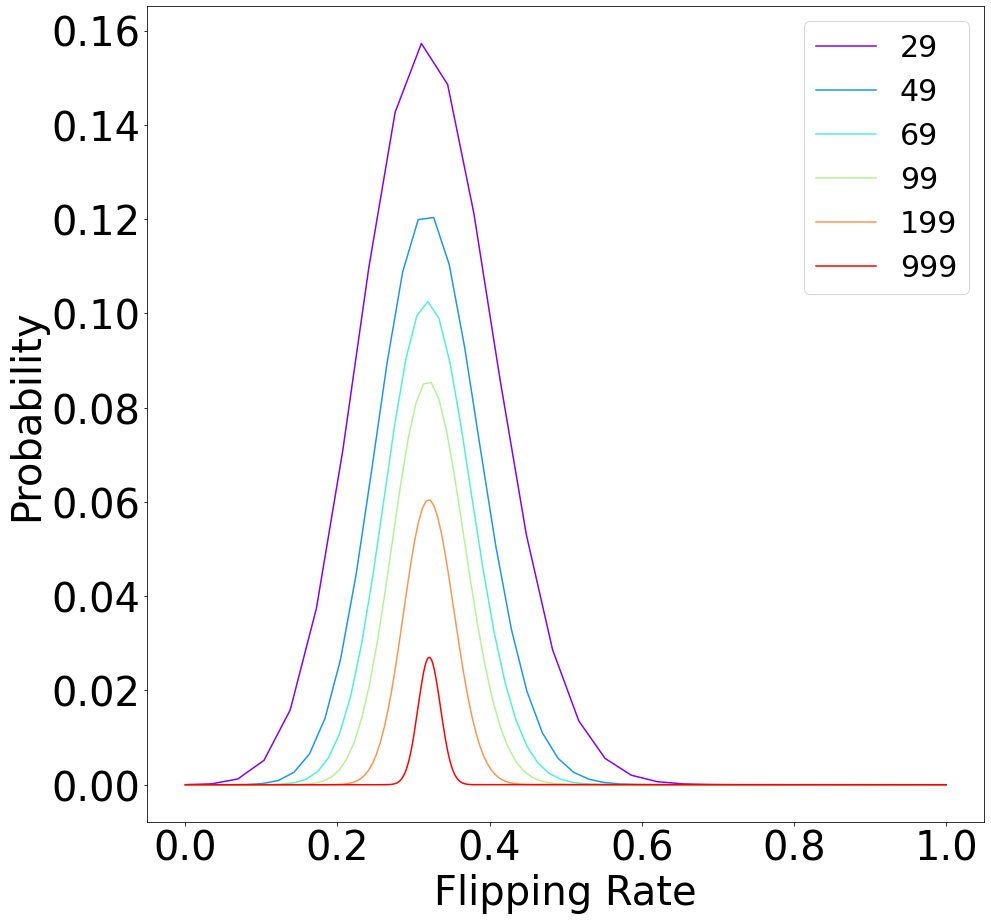}
    \caption{Example EMs with same $\epsilon=1.5$ (mean flipping probability $p=0.321$). As $n$ gets larger, the magnitude of the cumulative distribution accumulates over a narrower interval. Trend of degradation can be viewed.}
    \label{fig:my_label3}
\end{figure}
\begin{corollary} EM degrades to RR as $n \rightarrow \infty $.
\end{corollary}
\begin{proof}
Property \ref{prop6} states that the magnitude of the distribution would stack fully on one point if $n \rightarrow \infty $, which means the flip rate of labels would become fixed to  $p=\frac{1}{1+\exp(\frac{\epsilon}{2\Delta})}$. Therefore, the two-step EM actually becomes a one-step RR.
\end{proof}
The illustration showing how EM degrades to RR is shown in Fig. \ref{fig:my_label3}.
\begin{table*}[htbp]
\caption{mean and standard deviation of AUC/BER for 'AIRPLANE vs CAT' with different $\epsilon$, where $n=1000, 6000, 10000$}
\centering
\setlength{\tabcolsep}{0.8mm}{
\begin{tabular}{|c|c|c|c|c|c|c|c|c|c|c|c|c|c|c|}
\hline
\multirow{2}{*}{\diagbox[innerwidth=1cm]{$\epsilon$}{Loss}} & \multicolumn{2}{c|}{Barrier} & \multicolumn{2}{c|}{Sigmoid} & \multicolumn{2}{c|}{Unhinged} & \multicolumn{2}{c|}{Savage} & \multicolumn{2}{c|}{Logistic} & \multicolumn{2}{c|}{Squared} & \multicolumn{2}{c|}{Hinge}  \\
&AUC& BER& AUC & BER & AUC & BER & AUC &BER & AUC &BER & AUC & BER & AUC & BER\\
\hline
\multirow{3}{*}{0.1} & 52.5(7.7) & \textbf{51.5(13.2)} & 52.6(7.7) & 48.7(6.6) & 53.5(5.6) & 49.2(5.8) & \textbf{53.6(8.1)}& 51.1(3.8) & 51.0(6.7) & 50.5(1.9) & 52.5(6.0) & 50.7(4.3) & 52.4(5.9) & 51.0(4.0) \\
& \textbf{54.2(5.4)} & \textbf{59.5(5.3)} & 50.9(2.6)& 56.4(3.8) & 50.9(2.2) & 51.5(1.7) & 50.5(2.5)& 51.7(1.7) & 50.5(2.3) & 51.7(1.9) & 51.3(2.6) & 51.8(0.5) & 50.7(2.5) &  52.3(1.4) \\
& \textbf{55.1(4.2)} & \textbf{59.7(7.9)} & 50.9(1.9) & 58.2(3.7) & 52.3(1.5) & 51.2(1.6) & 51.7(1.8)& 50.2(1.2) & 51.7(1.4) & 50.5(1.1) & 51.4(1.7) & 50.6(1.4) & 51.0(1.7) & 50.4(1.0) \\
\hline
\multirow{3}{*}{0.5} & \textbf{76.0(7.1)} & \textbf{73.6(4.5)} & 61.1(3.8) & 65.3(8.2) & 58.7(5.1) & 56.9(4.0) & 60.0(4.8)& 56.5(3.4) & 61.4(5.6) & 55.3(6.5) & 60.4(3.8) & 55.3(2.7) & 59.4(2.6) & 58.1(4.0) \\
& \textbf{81.8(4.0)} & \textbf{80.5(2.0)} & 57.9(3.1)& 73.8(2.8) & 57.5(2.8) & 56.3(1.5) & 57.7(2.4)& 56.1(1.9) & 58.2(2.3) & 56.3(2.3) & 56.8(2.0) & 56.6(1.5) & 56.6(2.2) &  55.9(1.7) \\
& \textbf{86.3(1.4)} & \textbf{83.7(1.7)} & 58.7(1.5) & 79.4(2.3) & 58.9(1.7) & 56.8(2.9) & 58.6(1.7)& 56.1(2.0) & 58.7(2.1) & 56.4(2.6) & 58.5(2.0) & 56.0(1.9) & 58.2(1.6) & 56.4(3.1) \\
\hline
\multirow{3}{*}{1} & \textbf{84.4(2.4)} & \textbf{77.9(2.5)} & 67.7(3.2) & 76.1(3.3) & 65.8(3.4) & 60.7(4.3) & 65.3(4.9)& 61.6(4.6) & 66.2(3.8) & 61.7(2.3) & 63.5(3.9) & 60.5(4.5) & 67.6(1.7) &  60.9(6.0) \\
& \textbf{91.4(0.5)} & \textbf{83.2(2.0)} & 67.6(2.3)& 81.9(1.5) & 65.4(2.0) & 61.5(1.6) & 66.8(2.8)& 62.2(2.4) &65.7(3.3) & 62.7(1.0) & 65.3(2.3) & 61.6(2.2) & 65.8(2.1) &  61.9(1.8) \\
& \textbf{92.3(0.9)} & \textbf{86.2(0.6)} & 67.7(2.1) & 83.1(1.4) & 66.0(1.6) & 63.6(1.3) & 67.0(0.6)& 61.9(1.4) & 67.0(1.3) & 62.5(2.0) & 65.7(1.1) & 62.1(1.1) & 66.5(0.7) & 62.4(1.5) \\
\hline
\multirow{3}{*}{1.5} & \textbf{87.0(2.0)} & \textbf{80.2(1.8)} & 74.5(4.0) & 77.9(3.8) & 71.0(4.1) & 66.2(3.5) & 72.7(3.6)& 66.5(3.1) & 73.4(3.7) & 66.3(2.2) & 70.9(4.2) & 65.9(3.6)& 71.3(4.1) &  67.5(2.2) \\
& \textbf{92.9(0.5)} & \textbf{84.8(1.1)} & 74.7(1.4)& 83.9(0.6) & 72.5(1.3)& 67.8(1.9) & 73.4(2.4)& 67.4(1.1) & 73.6(2.7) & 67.6(1.7) & 72.8(2.4) & 68.0(1.4) & 72.3(1.3) &  66.7(1.6) \\
& \textbf{93.8(0.5)} & \textbf{86.9(0.4)} & 75.0(1.5) & 86.3(0.6) & 73.6(1.0) & 67.8(1.1) & 74.3(1.3)& 66.9(1.5) & 74.8(1.1) & 67.6(1.3) & 74.0(1.0) & 67.6(1.0) & 74.5(2.1) & 67.3(2.3) \\
\hline
\multirow{3}{*}{3} & \textbf{87.8(1.4)} & 81.4(1.7) & 85.9(2.0) & \textbf{83.1(1.6)} & 84.1(3.0) & 80.0(2.2) & 84.7(2.8)& 76.5(2.5) & 85.3(2.3) & 77.0(3.0) & 84.4(3.3) & 76.7(1.8) & 84.6(3.1) &  78.8(2.9) \\
& \textbf{94.1(0.3)} & 86.7(0.7) & 87.5(1.4)& \textbf{88.2(0.9)} & 86.2(1.0) & 81.1(1.4) & 87.2(1.3)& 80.1(1.0) & 87.4(1.0) & 79.2(1.5) & 87.6(1.0) & 79.9(1.5) & 86.3(0.8) &  81.2(0.9) \\
& \textbf{94.2(0.5)} & 88.0(0.5) & 89.2(1.0) & \textbf{89.6(0.4)} & 87.7(0.7) & 82.9(1.3) & 87.7(1.2)& 79.8(1.3) & 88.0(1.0) & 79.8(1.2) & 88.6(0.6) & 80.9(0.9) & 87.4(0.8) & 82.7(1.5) \\
\hline
\multirow{3}{*}{5} & 88.8(0.8) & 81.3(1.1) & 89.3(1.5) & \textbf{83.9(1.4)} & 89.8(1.7) & 82.9(2.5) & 89.7(1.5)& 83.4(2.4) & 90.1(1.5) & 82.3(2.6) & \textbf{90.4(1.6)} & \textbf{83.9(1.3)} & 89.8(2.2) &  83.6(1.3) \\
& 94.0(0.4) & 86.8(0.6) & 94.3(0.3)& \textbf{90.7(0.4)} & 93.8(0.9) & 88.9(0.9) & 94.1(0.6)& 88.7(0.8) & 94.4(0.7) & 88.0(0.9) & \textbf{95.2(0.4)} & 88.9(0.7) & 93.7(0.6) &  88.6(0.9) \\
& 94.2(0.3) & 88.8(0.5) & 94.6(0.6) & \textbf{91.3(0.4)} & 94.3(0.4) & 90.6(0.7) & 94.6(0.3)& 89.8(0.7) &94.7(0.5) & 88.6(0.6) & \textbf{95.4(0.3)} & 89.8(0.6) & 94.1(0.4) &  90.0(0.5) \\
\hline
\multirow{3}{*}{7} &  88.3(1.2) & 82.8(0.8) & 90.9(1.3) & 84.0(1.5) & 90.9(1.0) & 85.2(1.1) & 90.9(1.2)& 85.1(1.6) & 91.2(1.4)& 84.7(1.2) & \textbf{91.7(1.3)} & \textbf{85.8(0.9)} & 91.0(1.2) &  84.3(0.8) \\
& 93.8(0.7) & 87.2(0.5) & 95.7(0.4)& 90.8(0.5) &95.7(0.6) & 90.7(0.9) & 95.8(0.4)& 90.9(0.6) & 95.9(0.7) & 90.6(0.9) & \textbf{96.5(0.5)} & \textbf{91.1(0.4)} & 95.9(0.4) &  90.4(0.7) \\
& 94.7(0.4) & 88.9(0.3) & 96.7(0.3) & 91.8(0.3) & 96.5(0.3) & 91.5(0.6) & 96.6(0.3)& 92.0(0.6) & 96.5(0.2) & 91.0(0.5) & \textbf{97.2(0.2)} & \textbf{92.2(0.4)} & 96.6(0.4) & 91.7(0.5) \\
\hline
\end{tabular}}
\end{table*}
\begin{figure*}
\centering
\setkeys{Gin}{width=0.33\linewidth}
\subfigure{\includegraphics{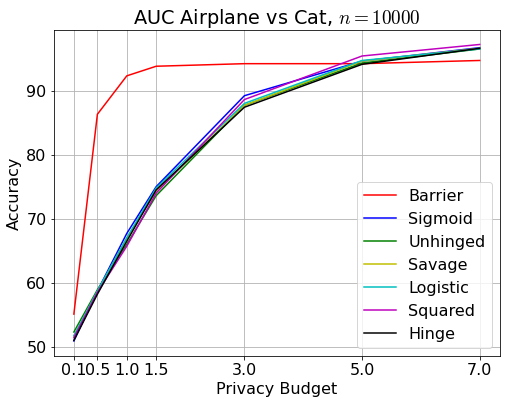}}
\hfill
\subfigure{\includegraphics{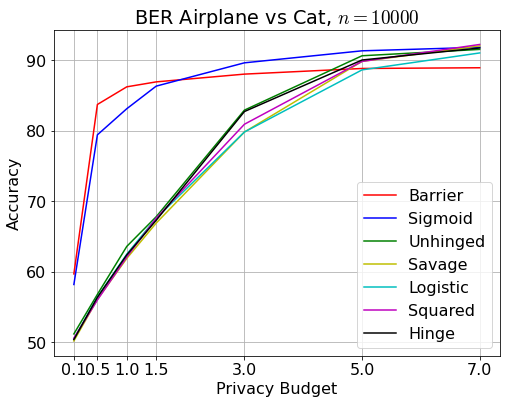}}
\hfill
\subfigure{\includegraphics{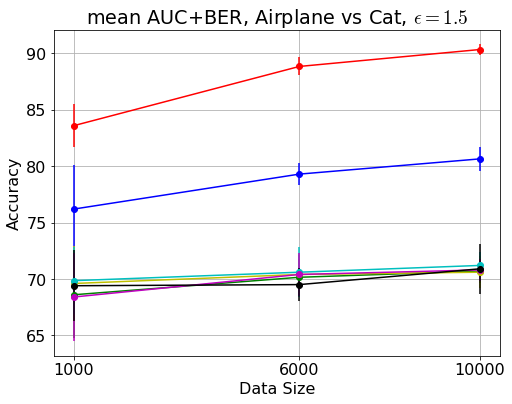}}
\caption{AUC score and balanced accuracy (BER) correspond to binary classification with varying privacy budgets and data sizes. Experiments were conducted 10 times. Means and standard deviations are displayed.}
\label{fig:U}
\end{figure*}

\section{Experiment}
Experiments were conducted on the CIFAR-10 dataset with a convolutional neural network (CNN). The AUC score and balanced accuracy (BER) were used for the evaluation. Data sizes were majorly issued to full scale (10000 for training, 2000 for testing), and $10\%$ scale (1000 for training, 200 for testing), see Table III, \ref{table12}ppendix. Privacy budgets were chosen in a common range. Labels were binary-categorized, and were processed by EM beforehand. The experiment was conducted 10 times, each consisted of 50 epochs. Other commonly used losses are comparing groups. We set the barrier hinge loss parameters as $b=200$ and $r=50$.

Overall, barrier hinge loss presents higher accuracies on medium small privacy parameters (e.g., $\epsilon \in [0.1,3]$), which results from the robustness of symmetry on the noise condition \cite{charoenphakdee2019symmetric}. It also presents acceptable accuracies for moderate privacy (e.g., $\epsilon \geq 0.5$), which could outstrip other losses for real applications. One plausible reason for the slightly lower performances of extremely clean data could be the value of symmetric range $r$. One may tune $b$ and $r$ to fit specific needs. Leveraging EM to process the data also results in higher standard deviations of accuracies than RR \cite{charoenphakdee2019symmetric}, which is reasoned from the flexibility of the flipping rate as shown in Fig. \ref{fig:my_label3}. Further, as $n$ gets larger, standard deviation decreases Fig. \ref{fig:U}. Training with small $n$ may result in underfitting, which could be explained by the convergence speed of each loss in \cite{charoenphakdee2019symmetric}. We also observe that the standard deviation of barrier hinge loss converges faster than other loss w.r.t. the increasing of $\epsilon$. While sigmoid function presents better convergence rate w.r.t. $n$, sigmoid function also presents its second optimum among all functions in our experiments, which re-confirms its success in recent years. Although in our analysis, convergence rates (e.g., of standard deviations) are not studied, it could be interesting for a further research, because the strength of oscillating of data accuracy is one elementary control corresponding to purposes.

\section{Conclusion}
We explored the performance of barrier hinge loss on EM processed private data. The experimental results show the robustness of this loss over other commonly used loss functions. Although our results were only derived from AUC and BER evaluations, the experiments can be evaluated with other standards easily. Moreover, we analyzed the properties of EM and provided numerical guidance of setting private parameters. Further, it was found that with the common definition of utility, EM would degrade to RR when the data size increases to infinity. This phenomenon is not fully understood, and is worth further investigation.

\section*{Acknowledgement}
This work is partially supported by JSPS Grant-in-Aid for Transformative Research Areas A grant number JP21H05845.
The authors would like to thank Prof. Reiji Suda for hosting Jing Bi at his lab during the course of this resarch.

\bibliography{myref} 

\begin{thebibliography}{10}
\providecommand{\url}[1]{#1}
\csname url@samestyle\endcsname
\providecommand{\newblock}{\relax}
\providecommand{\bibinfo}[2]{#2}
\providecommand{\BIBentrySTDinterwordspacing}{\spaceskip=0pt\relax}
\providecommand{\BIBentryALTinterwordstretchfactor}{4}
\providecommand{\BIBentryALTinterwordspacing}{\spaceskip=\fontdimen2\font plus
\BIBentryALTinterwordstretchfactor\fontdimen3\font minus
  \fontdimen4\font\relax}
\providecommand{\BIBforeignlanguage}[2]{{%
\expandafter\ifx\csname l@#1\endcsname\relax
\typeout{** WARNING: IEEEtran.bst: No hyphenation pattern has been}%
\typeout{** loaded for the language `#1'. Using the pattern for}%
\typeout{** the default language instead.}%
\else
\language=\csname l@#1\endcsname
\fi
#2}}
\providecommand{\BIBdecl}{\relax}
\BIBdecl

\bibitem{Sweene02}
L.~Sweeney, ``k-anonymity: {A} model for protecting privacy,''
  \emph{International Journal of Uncertainty, Fuzziness and Knowledge-Based
  Systems}, vol.~10, no.~5, pp. 557--570, 2002.

\bibitem{machanavajjhala2007diversity}
A.~Machanavajjhala, D.~Kifer, J.~Gehrke, and M.~Venkitasubramaniam,
  ``l-diversity: Privacy beyond k-anonymity,'' \emph{ACM Transactions on
  Knowledge Discovery from Data (TKDD)}, vol.~1, no.~1, 2007.

\bibitem{Dwork06}
C.~Dwork, ``Differential privacy,'' in \emph{ICALP 2006}, 2006, pp. 1--12.

\bibitem{PETS22}
X.~Tang, M.~Nasr, S.~Mahloujifar, V.~Shejwalkar, L.~Song, A.~Houmansadr, and
  P.~Mittal, ``Machine learning with differentially private labels: Mechanisms
  and frameworks,'' in \emph{PETS 2022}, 2022.

\bibitem{charoenphakdee2019symmetric}
N.~Charoenphakdee, J.~Lee, and M.~Sugiyama, ``On symmetric losses for learning
  from corrupted labels,'' in \emph{ICML 2019}, 2019, pp. 961--970.

\bibitem{warner1965randomized}
S.~L. Warner, ``Randomized response: A survey technique for eliminating evasive
  answer bias,'' \emph{Journal of the American Statistical Association},
  vol.~60, no. 309, pp. 63--69, 1965.

\bibitem{Yue2016}
Y.~Wang, X.~Wu, and D.~Hu, ``Using randomized response for differential privacy
  preserving data collection,'' in \emph{EDBT/ICDT 2016}, 2016.

\bibitem{mcsherry2007mechanism}
F.~McSherry and K.~Talwar, ``Mechanism design via differential privacy,'' in
  \emph{FOCS 2007}, 2007, pp. 94--103.

\bibitem{b1}
L.~Roohi, B.~I. Rubinstein, and V.~Teague, ``Differentially-private two-party
  egocentric betweenness centrality,'' in \emph{IEEE INFOCOM 2019}, 2019, pp.
  2233--2241.

\bibitem{Dwork2014}
C.~Dwork and A.~Roth, ``The algorithmic foundations of differential privacy.''
  \emph{Foundations and Trends in Theoretical Computer Science}, vol.~9, no.
  3-4, pp. 211--407, 2014.

\bibitem{Mironov2017}
I.~Mironov, ``Rényi differential privacy,'' in \emph{CSF 2017}, 2017, pp.
  263--275.

\bibitem{Abadi2016}
M.~Abadi, A.~Chu, I.~Goodfellow, H.~B. McMahan, I.~Mironov, K.~Talwar, and
  L.~Zhang, ``Deep learning with differential privacy,'' in \emph{SIGACT 2016},
  2016.

\bibitem{Narasimhan2016}
H.~Narasimhan and S.~Agarwal, ``Support vector algorithms for optimizing the
  partial area under the {ROC} curve,'' \emph{Neural Computation}, vol.~29, pp.
  1919--–1963, 2017.

\bibitem{Feldman2014}
M.~Feldman, S.~Friedler, J.~Moeller, C.~Scheidegger, and S.~Venkatasubramanian,
  ``Certifying and removing disparate impact,'' in \emph{KDD 2015}, 2015.

\bibitem{Zhao2020}
H.~Zhao, A.~Coston, T.~Adel, and G.~J. Gordon, ``Conditional learning of fair
  representations,'' in \emph{ICLR 2020}, 2020.

\bibitem{Bartlett2006}
P.~Bartlett, M.~Jordan, and J.~McAuliffe, ``Convexity, classification, and risk
  bounds,'' \emph{Journal of the American Statistical Association}, vol. 101,
  pp. 138--156, 02 2006.

\bibitem{Plessis2014}
M.~C. du~Plessis, G.~Niu, and M.~Sugiyama, ``Analysis of learning from positive
  and unlabeled data,'' in \emph{NeurIPS 2014}, 2014, p. 703–711.

\bibitem{Ghosh2014}
A.~Ghosh, N.~Manwani, and P.~Sastry, ``Making risk minimization tolerant to
  label noise,'' \emph{Neurocomputing}, vol. 160, pp. 93--107, 2015.

\bibitem{Jordan2019}
J.~Awan, A.~Kenney, M.~Reimherr, and A.~Slavkovi{\'c}, ``Benefits and pitfalls
  of the exponential mechanism with applications to {Hilbert spaces and
  functional PCA},'' in \emph{ICML 2019}, 2019, pp. 374--384.

\bibitem{HighDim2018}
R.~Vershynin, \emph{High-Dimensional Probability: An Introduction with
  Applications in Data Science}, ser. Cambridge Series in Statistical and
  Probabilistic Mathematics.\hskip 1em plus 0.5em minus 0.4em\relax {Cambridge
  University Press}, no.~47.

\bibitem{Ajoy2018}
A.~Thamattoor, ``Normal limit of the binomial via the discrete derivative,''
  \emph{The College Mathematics Journal}, vol.~49, pp. 216--217, 05 2018.

\end{thebibliography}


\begin{thebibliography}{00}
\bibitem{b1} 
\end{thebibliography}
\bibliographystyle{IEEEtran}

\newpage
\begin{appendix}

\noindent\textbf{Property 1.} Truncated binomial distribution $S(n,j)$ is monotonically decreasing w.r.t. $n$, while $j\in [0,n]$ is fixed
\label{app1}

\begin{proof} 
The third equality is obtained from considering $I_{n,p}$ as the sum of $n$ trials, while $I_{n+1,p}$ is obtained from adding one additional trial to $I_{n,p}$.
\begin{eqnarray*}
& & S(n+1,j)\\
&=& \Pr\left[ I_{n+1,p} \leq j\right]\\
& = & \Pr\left[I_{n+1,p} \leq j|I_{n,p} < j \right]\Pr\left[I_{n,p} < j\right]\\
& & +\Pr\left[I_{n+1,p} \leq j|I_{n,p} = j \right]\Pr\left[I_{n,p} = j\right]\\
& & +\Pr\left[I_{n+1,p} \leq j|I_{n,p} \geq j+1 \right]\Pr\left[I_{n,p} \geq j+1\right]\\
& = & 1\cdot \sum_{i=0}^{j-1} {n \choose i}p^i(1-p)^{n-i} \\
& & +(1-p){n \choose j}p^j(1-p)^{n-j}+0 \\
& = & \sum_{i=0}^{j} {n \choose i}p^i(1-p)^{n-i} - p{n \choose j}p^j(1-p)^{n-j}\\
& \leq &  \sum_{i=0}^{j} {n \choose i}p^i(1-p)^{n-i}\\
&=& S(n,j).
\end{eqnarray*}

\end{proof}

\noindent\textbf{Property 2.} Truncated binomial distribution $S(n,n-k)$ is monotonically increasing w.r.t. $n$, while $k \in [0,n_0]$ and $n_0$ are fixed.

\begin{proof}
Take $n-k=j$,
\begin{eqnarray*}
& & S(n+1,n+1-k)\\
&=& S(n+1,j+1) \\
& = &\Pr\left[I_{n+1,p} \leq j+1\right]\\
& = &\Pr\left[I_{n+1,p} \leq j+1|I_{n,p} \leq j \right]\Pr\left[I_{n,p} \leq j\right]\\
& &+\Pr\left[I_{n+1,p} \leq j+1|I_{n,p} = j+1 \right]\Pr\left[I_{n,p} = j+1\right]\\
& &+\Pr\left[I_{n+1,p} \leq j+1|I_{n,p} > j+1\right]\Pr\left[I_{n,p} > j+1\right]\\
& = &1\cdot S(n,j) + (1-p){n \choose j+1}p^{j+1}(1-p)^{n-j-1}+0 \\
& \geq & S(n,j)\\
&=& S(n,n-k).
\end{eqnarray*}
\end{proof}
\noindent\textbf{Remark:} Since above increments of the truncated points have the extreme ($0$ or $1$) proposition to $n$, the above two monotonicities can be obtained without considering the distribution probability $p$. Therefore, Property 2 could also be deduced by $1-S'(n,k)$ from Property 1, while $S'(n,k)$ has the probability $1-p$. There may exist other methods of proofs--for example, taking derivatives and integrals. However, because of the complexity, we do not use them here.\\\\
\noindent\textbf{Property 3.}
Truncated binomial distribution $S(n,\lceil\frac{n}{2} \rceil+k)$ is almost monotonic w.r.t. $n$, except for when $n$ changes between even and odd––this is while $k \in [-\lceil\frac{n_0}{2} \rceil, \lfloor \frac{n_0}{2} \rfloor]$, and $n_0$ is fixed. Moreover, if $p \leq \frac{n-\lceil\frac{n}{2} \rceil-k}{n+1} $, $S$ is almost monotonically increasing; and is decreasing if $p \geq \frac{n-\lceil\frac{n}{2} \rceil-k}{n+1} $, respectively. 

\begin{proof}
Due to the existence of the term $\lceil\frac{n}{2} \rceil$, we have to consider the parity of $n$. Taking $\lceil\frac{n}{2} \rceil+k = j$ for simplicity.\\
If $n$ is odd and $n + 1$ is even, then 
$S(n+1,\lceil\frac{n+1}{2} \rceil+k)= S(n+1,j+1) \geq S(n,j)$.
If $n$ is even and $n + 1$ is odd, then 
$S(n+1,\lceil\frac{n+1}{2} \rceil+k)= S(n+1,j) \leq S(n,j)$.
Now, considering $n$ increases as always being odd numbers (or even respectively), then:
\begin{eqnarray*}
&&S(n+2,\lceil\frac{n+2}{2} \rceil+k)\\
&=&  S(n+2,j+1)\\
&=&\Pr\left[I_{n+2,p} \leq j+1\right]\\
&=&\Pr\left[I_{n+2,p} \leq j+1|I_{n,p} \leq j-1\right]\Pr\left[I_{n,p} \leq j-1\right]\\
&& + \Pr\left[I_{n+2,p} \leq j+1|I_{n,p} = j\right]\Pr\left[I_{n,p} =j\right] \\
&&+\Pr\left[I_{n+2,p} \leq j+1|I_{n,p} = j+1\right]\Pr\left[I_{n,p} =j+1\right] \\
&&+\Pr\left[I_{n+2,p} \leq j+1|I_{n,p} > j+1\right]\Pr\left[I_{n,p} >j+1\right]\\
&=& 1 \cdot S(n,j-1) \\
&&+ (2p(1-p)+(1-p)^2){n \choose j}p^j(1-p)^{n-j} \\
&&+ (1-p)^2 {n\choose j+1}p^{j+1}(1-p)^{n-j-1} +0 \\
&=&S(n,j)-{n \choose j}p^j(1-p)^{n-j} \\
&&+ (2p(1-p)+(1-p)^2){n\choose j}p^j(1-p)^{n-j}\\
&&+ (1-p)^2 {n\choose j+1}p^{j+1}(1-p)^{n - j - 1}\\
&=&S(n,j)+\left[-1+2p(1-p)+(1-p)^2 \right.\\
&&\left. + p(1-p)\frac{n-j}{j+1}\right] {n \choose j}p^j(1-p)^{n-j}\\
&=&S(n,j)+\left[-p^2 + (p-p^2)\frac{n-j}{j+1}\right] {n \choose j}p^j(1-p)^{n-j}\\
&=&S(n,j)+\left[p(\frac{n-j}{j+1}-\frac{n+1}{j+1}p)\right] {n \choose j}p^j(1-p)^{n-j}.\\
\end{eqnarray*}
Hence, if $$\frac{n-j}{j+1}-\frac{n+1}{j+1}p \geq 0,$$ i.e., 
$$ p \leq \frac{n-j}{n+1},$$
then $$S(n+2,j+1) \geq S(n,j);$$ or otherwise, $$S(n+2,j+1) \leq S(n,j).$$ 
\end{proof}

\noindent\textbf{Property 4.}
The truncated binomial distribution $\Pr\left[I_{n,p} \leq j\right]$ is monotonically decreasing w.r.t. probability $p$. Moreover, if $p=0$, then $\Pr\left[I_{n,p} \leq j\right]=1$ constantly; if $p=1$, then $\Pr\left[I_{n,p} \leq j\right]=0$.
\begin{proof} 
It is easy to prove the statement the extreme cases ($p=0$ or $1$). 
For $0 < p < 1$, recall that the truncated binomial distribution is given by 
$$\Pr\left[I_{n,p} \leq j\right]=\sum_{i=0}^{j} {n\choose i}p^i(1-p)^{n-i}.$$
Now, for any $p\in [0,1)$, consider a small $\delta > 0$ such that $p+\delta \in (0,1]$, and compare the ratios of two neighboring terms with different given probabilities:
\begin{eqnarray*}
&&\frac{{n\choose i+1} p^{i+1}(1-p)^{n-i-1}}{{n\choose i} p^{i}(1-p)^{n-i}}=\frac{{n\choose i+1} p}{{n\choose i} (1-p)} \\
&<& \frac{{n\choose i+1} (p+\delta)}{{n\choose i} (1-p-\delta)}= \frac{{n\choose i+1} (p+\delta)^{i+1}(1-p-\delta)^{n-i-1}}{{n\choose i} (p+\delta)^{i}(1-p-\delta)^{n-i}}.
\end{eqnarray*}
This means that terms in the truncated binomial distribution with a larger probability increase faster or decrease slower. The trend is illustrated as in Fig. \ref{fig:my_label2}.

\begin{figure}[htbp]
    \centering
    \includegraphics[scale=0.53]{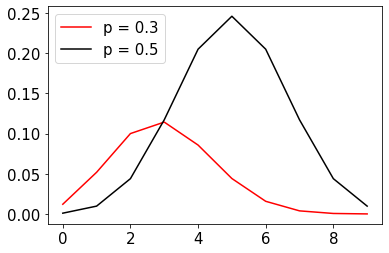}
    \caption{Example trends of ${n\choose j}p^j(1-p)^{n-j}$ for $j\in[n=9]$ w.r.t. $p=0.3$ and $p=0.5$. The trends show that an interchange of function values w.r.t. two probabilities indeed exists.}
    \label{fig:my_label2}
\end{figure}

Let $$R_j := \frac{{n\choose j}p^j(1-p)^{n-j}}{{n\choose j}(p+\delta)^j(1-p-\delta)^{n-j}}.$$ It is clear that $R_0 > 1$. Since ${n\choose i}(p+\delta)^i(1-p-\delta)^{n-i}$ increases faster or decreases slower, then for large enough $j'$, we would have $R_j < 1$ for $j > j'$. In this case, we use subtraction:
\begin{eqnarray*}
&&\Pr\left[I_{n,p} \leq j\right]\\
&=& 1- \sum_{i=j+1}^{n} {n\choose i}p^i(1-p)^{n-i}\\
&>&1- \sum_{i=j+1}^{n} {n\choose i}(p+\delta)^i(1-p-\delta)^{n-i}\\
&=& \Pr\left[I_{n,p+\delta} \leq j\right].
\end{eqnarray*}
In contrast, when $j \leq j'$, we have 
\begin{eqnarray*}
&&\Pr\left[I_{n,p} \leq j\right]\\
&=&\sum_{i=0}^{j} {n\choose i}p^i(1-p)^{n-i}\\
&>& \sum_{i=0}^{j} {n\choose i}(p+\delta)^i(1-p-\delta)^{n-i}\\
&=& \Pr\left[I_{n,p+\delta} \leq j\right].
\end{eqnarray*}
\end{proof}
\noindent\textbf{Corollary 3.} To ensure a $99.9\%$ probability of having a flipping rate that is no greater than $\frac{1}{2}$, we need a privacy budget to be no less than $2\Delta \log\frac{1+2\sqrt{(1.5\log10)/n}}{1-2\sqrt{(1.5\log10)/n}}$.
\begin{proof} Following the steps in proof of Proposition \ref{pros2}, we use Hoeffding's inequality \cite{HighDim2018}:

$$S(n,\lceil\frac{n}{2} \rceil) \geq 1 - \exp\left(-\frac{2(n(\frac{1}{2}-p))^2}{n}\right)\geq 0.999 \\$$
then 
$$\exp\left(-2n(\frac{1}{2}-p)^2\right) \leq 0.001$$
$$\epsilon \geq 2\Delta \log\frac{1+2\sqrt{(1.5\log10)/n}}{1-2\sqrt{(1.5\log10)/n}}$$
For instance, if $n=100$, then $\epsilon \geq 1.562 $; if $n=1000$, then $\epsilon \geq 0.472 $; if $n=10000$, then $\epsilon \geq 0.149 $.
\end{proof}

\noindent\textbf{Property 5.}
The truncated binomial distribution $S(n,\lceil\frac{n}{2} \rceil)$ approaches $1$ if $p <\frac{1}{2}$, $0$ if $p >\frac{1}{2}$, and $\frac{1}{2}$ if $p=\frac{1}{2}$, as $n \rightarrow \infty $.
\begin{proof}
If $p <\frac{1}{2}$, then
$$\lceil\frac{n}{2} \rceil-np \geq n(\frac{1}{2}-p)>0.$$
Taking $\frac{1}{2}-p=c$ and following similar steps in Proposition \ref{pros2} with carrying the limit:
\begin{eqnarray*}
&&\lim_{n \rightarrow \infty}S(n,\lceil\frac{n}{2} \rceil)\\
&=& \lim_{n \rightarrow \infty}\Pr\left[I_{n,p} \leq \lceil\frac{n}{2} \rceil\right]\\
&=& \lim_{n \rightarrow \infty}\Pr\left[\sum_{i=1}^{n}X_i \leq \lceil\frac{n}{2} \rceil\right] \\
&=&1 -\lim_{n \rightarrow \infty}\Pr\left[\sum_{i=1}^{n}X_i > \lceil\frac{n}{2} \rceil\right] \\
&=& 1 -\lim_{n \rightarrow \infty}\Pr\left[\sum_{i=1}^{n}X_i -np> \lceil\frac{n}{2} \rceil-np\right] \\
&\geq& 1 -\lim_{n \rightarrow \infty}\Pr\left[\sum_{i=1}^{n}X_i -np> nc\right] \\
&\geq& 1 - \lim_{n \rightarrow \infty}\exp\left(-\frac{2(nc)^2}{n}\right)\\
&=&1.
\end{eqnarray*}
If $p > \frac{1}{2}$, then 
$$\lceil\frac{n}{2} \rceil-np \leq n(\frac{1}{2}-p)+\frac{1}{2} <0.$$ 
Thus,
\begin{eqnarray*}
&&\lim_{n \rightarrow \infty}\Pr\left[\sum_{i=1}^{n}X_i \leq \lceil\frac{n}{2} \rceil\right]\\
&=& \lim_{n \rightarrow \infty}\Pr\left[\sum_{i=1}^{n}X_i -np \leq \lceil\frac{n}{2} \rceil-np\right] \\
&\leq& \lim_{n \rightarrow \infty}\Pr\left[\sum_{i=1}^{n}X_i -np \leq nc+\frac{1}{2}\right] \\
&\leq& \lim_{n \rightarrow \infty}\exp\left(-\frac{2(nc+\frac{1}{2})^2}{n}\right)\\
&=& 0.
\end{eqnarray*}
If $p = \frac{1}{2}$, it is easy to check.
\end{proof}

\noindent\textbf{Property 6.} For any $\delta \in (0,1]$, the truncated binomial distribution $S(n,n(p-\delta),n(p+\delta))$ approaches $1$ as $n \rightarrow \infty $.
\begin{proof}
\begin{eqnarray*}
&&\lim_{n \rightarrow \infty}S(n,n(p-\delta),n(p+\delta))\\
&=&\lim_{n \rightarrow \infty}\Pr\left[n(p-\delta) \leq I_{n,p} \leq n(p+\delta)\right]\\
&=& \lim_{n \rightarrow \infty}\Pr\left[n(p-\delta) \leq \sum_{i=1}^{n}X_i \leq n(p+\delta)\right]\\
&=& \lim_{n \rightarrow \infty}\Pr\left[-n\delta \leq \sum_{i=1}^{n}X_i -np \leq n\delta\right]\\
&=& 1 -\lim_{n \rightarrow \infty}\Pr\left[\sum_{i=1}^{n}X_i-np< -n\delta \right]\\
&&-\lim_{n \rightarrow \infty}\Pr\left[\sum_{i=1}^{n}X_i -np>n\delta \right] \\
&\geq& 1 - 2\lim_{n \rightarrow \infty}\exp\left(-\frac{2(n\delta)^2}{n}\right)\\
&=&1.\\
\end{eqnarray*}
\end{proof}
\newpage
\begin{table*}[ht]
\caption{mean and standard deviation of AUC/BER, where $n=10000$, $\epsilon=0.1,0.5,1,1.5,3,5,7$}
\centering
\setlength{\tabcolsep}{0.75mm}{
\begin{tabular}{|c|c|c|c|c|c|c|c|c|c|c|c|c|c|c|}
\hline
\multirow{2}{*}{Dataset} & \multicolumn{2}{c|}{Barrier} & \multicolumn{2}{c|}{Sigmoid} & \multicolumn{2}{c|}{Unhinged} & \multicolumn{2}{c|}{Savage} & \multicolumn{2}{c|}{Logistic} & \multicolumn{2}{c|}{Squared} & \multicolumn{2}{c|}{Hinge}  \\
&AUC& BER& AUC & BER & AUC & BER & AUC &BER & AUC &BER & AUC & BER & AUC & BER\\
\hline
\multirow{7}{*}{automobile} & \textbf{55.3(4.3)} & \textbf{56.5(4.9)} & 50.5(1.3) & 54.4(2.7) & 51.5(2.3) &51.1(1.1) &  51.9(1.8)& 50.6(1.3) & 51.4(1.5) & 51.2(0.9) & 52.2(1.8) & 51.1(1.4) & 51.9(1.4) & 50.7(1.0) \\
& \textbf{82.7(2.5)} &\textbf{79.0(1.3)} & 59.7(2.4) & 73.8(1.4) & 58.4(1.6) & 56.5(1.2) & 58.8(1.5)& 56.2(1.6) & 59.2(1.9) & 55.8(1.7) & 58.0(1.8) & 56.3(1.2) & 57.9(2.6) & 55.6(1.3) \\
& \textbf{92.0(0.9)} & \textbf{85.1(0.4)} & 67.7(2.7) & 81.7(0.9) & 66.3(1.5) & 62.4(1.4) & 66.6(1.8)& 62.0(1.6) & 66.7(1.5) & 62.4(1.7) & 66.9(1.4) & 62.4(1.0) & 66.8(0.8) &  62.1(1.9) \\
& \textbf{94.0(0.7)} & \textbf{85.6(0.6)} & 74.7(2.5) & 85.7(0.7) & 74.7(1.8) & 67.8(1.5) &  74.7(1.1)& 68.0(2.1) & 74.6(1.1) & 68.0(1.5) & 74.6(1.5) & 68.0(1.1) & 74.1(1.8) & 67.8(1.7)\\
& \textbf{95.3(0.4)} & 87.3(0.4) & 88.5(1.2) & \textbf{90.5(0.5)} & 87.6(0.8) & 80.9(1.2) & 88.6(0.9)& 80.2(1.0) & 88.3(0.8) & 80.2(0.9) & 88.7(0.7) &81.1(1.0) & 88.1(0.9) &  81.0(1.1) \\
& 95.9(0.4) & 87.8(0.5) & 95.2(0.7) & \textbf{92.4(0.4)} & 95.1(0.4) & 89.6(0.7) &95.2(0.4)& 89.9(0.7) & 95.1(0.3) & 88.4(0.8) & \textbf{96.1(0.3)} & 90.3(0.5) & 95.2(0.3) & 89.3(0.6) \\
& 96.0(0.2) & 87.7(0.4) & 97.3(0.3) & 92.7(0.3) & 97.0(0.2) & 91.7(0.4) & 97.2(0.2)& 92.8(0.5) & 97.3(0.1) & 92.4(0.9) & \textbf{97.7(0.3)} & \textbf{93.1(0.5)} & 97.1(0.3) & 91.8(0.4)\\
\hline
\multirow{7}{*}{bird} & \textbf{54.8(4.9)} & \textbf{59.9(3.6)} & 51.3(2.4) & 54.7(1.7) & 50.7(1.5) & 51.1(1.2) & 51.3(1.5)& 51.1(1.2) & 51.3(1.9) & 51.2(1.6) & 51.2(1.7) & 51.0(1.2) & 51.4(2.0) & 51.4(1.2) \\
& \textbf{81.3(3.5)} & \textbf{78.1(2.7)} & 58.4(2.6) & 73.4(2.5) & 57.9(1.2) & 55.2(0.9) &57.1(1.2)&56.3(1.5) &58.9(1.5) & 55.6(1.5) & 56.6(1.8) & 55.3(0.8) & 57.4(2.0) & 55.7(1.3) \\
& \textbf{89.1(0.7)} & \textbf{81.3(1.5)} & 65.9(1.5) & 79.3(1.2) & 64.5(1.8) & 60.2(2.0) &64.9(1.6)& 61.7(1.1) &64.6(1.9) & 60.2(1.3) & 65.1(1.6) & 61.1(1.4) & 64.9(1.3) & 59.7(1.6) \\
& \textbf{90.5(0.5)} & \textbf{82.2(0.9)} & 72.2(1.9) & 82.1(1.1) & 70.3(1.3) & 65.2(1.4) &71.6(1.6)&65.8(1.4) & 71.1(1.4) & 64.9(1.1) & 71.3(0.8) & 64.7(0.9) & 70.9(1.4) & 65.5(1.3) \\
& \textbf{91.2(0.5)} & 83.5(0.5) & 85.7(0.8) & \textbf{85.9(0.7)} & 83.9(0.8) & 78.7(1.2) & 84.3(1.1)&77.4(1.1) & 84.4(1.1) & 76.4(1.3) & 85.0(0.8) & 77.0(0.9) & 83.9(1.0) & 78.1(1.3)\\
& 91.4(0.3) & 83.9(0.3) & 91.6(0.7) & \textbf{87.7(0.4)} & 91.4(0.8) & 85.1(0.5) &91.7(0.5)&85.0(0.8) & 91.5(0.5) &83.6(0.6) & \textbf{92.2(0.6)} & 85.0(0.6) & 91.1(0.6) &  85.2(0.9) \\
& 91.6(0.3) & 84.0(0.4) & 94.0(0.4) & \textbf{87.8(0.5)} & 93.6(0.4) & 87.2(0.5) &93.8(0.3)& \textbf{87.8(0.6)} & 93.9(0.3) &86.8(0.9) & \textbf{94.2(0.5)} & 87.6(0.7) & 93.7(0.4) &  86.9(0.6) \\
\hline
\multirow{7}{*}{cat} & \textbf{55.1(4.2)} & \textbf{59.7(7.9)} & 50.9(1.9) & 58.2(3.7) & 52.3(1.5) & 51.2(1.6) & 51.7(1.8)& 50.2(1.2) & 51.7(1.4) & 50.5(1.1) & 51.4(1.7) & 50.6(1.4) & 51.0(1.7) & 50.4(1.0) \\
& \textbf{86.3(1.4)} & \textbf{83.7(1.7)} & 58.7(1.5) & 79.4(2.3) & 58.9(1.7) & 56.8(2.9) & 58.6(1.7)& 56.1(2.0) & 58.7(2.1) & 56.4(2.6) & 58.5(2.0) & 56.0(1.9) & 58.2(1.6) & 56.4(3.1) \\
& \textbf{92.3(0.9)} & \textbf{86.2(0.6)} & 67.7(2.1) & 83.1(1.4) & 66.0(1.6) & 63.6(1.3) & 67.0(0.6)& 61.9(1.4) & 67.0(1.3) & 62.5(2.0) & 65.7(1.1) & 62.1(1.1) & 66.5(0.7) & 62.4(1.5) \\
& \textbf{93.8(0.5)} & \textbf{86.9(0.4)} & 75.0(1.5) & 86.3(0.6) & 73.6(1.0) & 67.8(1.1) & 74.3(1.3)& 66.9(1.5) & 74.8(1.1) & 67.6(1.3) & 74.0(1.0) & 67.6(1.0) & 74.5(2.1) & 67.3(2.3) \\
& \textbf{94.2(0.5)} & 88.0(0.5) & 89.2(1.0) & \textbf{89.6(0.4)} & 87.7(0.7) & 82.9(1.3) & 87.7(1.2)& 79.8(1.3) & 88.0(1.0) & 79.8(1.2) & 88.6(0.6) & 80.9(0.9) & 87.4(0.8) & 82.7(1.5) \\
& 94.2(0.3) & 88.8(0.5) & 94.6(0.6) & \textbf{91.3(0.4)} & 94.3(0.4) & 90.6(0.7) & 94.6(0.3)& 89.8(0.7) &94.7(0.5) & 88.6(0.6) & \textbf{95.4(0.3)} & 89.8(0.6) & 94.1(0.4) &  90.0(0.5) \\
& 94.7(0.4) & 88.9(0.3) & 96.7(0.3) & 91.8(0.3) & 96.5(0.3) & 91.5(0.6) & 96.6(0.3)& 92.0(0.6) & 96.5(0.2) & 91.0(0.5) & \textbf{97.2(0.2)} & \textbf{92.2(0.4)} & 96.6(0.4) & 91.7(0.5) \\
\hline
\multirow{7}{*}{deer} & \textbf{57.9(5.3)} & \textbf{62.6(4.6)} & 52.1(2.5) & 59.2(5.0) &52.4(2.4) & 52.0(1.2) & 51.9(1.1)& 51.7(1.0) & 52.1(1.0) & 51.8(1.5) & 52.4(1.6) & 51.6(1.1) & 51.9(1.8) &  51.2(1.4) \\
& \textbf{89.2(1.1)} & \textbf{84.4(1.2)} & 60.1(2.1) & 79.6(2.1) & 59.2(1.9) & 56.0(1.1) & 58.0(1.4)& 56.1(1.1) & 58.4(2.7) & 56.4(2.0) & 58.2(2.0) & 55.7(0.9) & 58.0(1.6) &  55.9(1.5) \\
& \textbf{92.7(0.9)} & \textbf{86.9(0.6)} & 65.9(1.5) & 83.4(1.1) & 65.1(1.7) & 62.2(1.2) & 66.3(1.8)& 61.9(1.4) & 66.1(1.2) & 62.1(1.4) & 65.6(1.4) & 61.5(1.5) & 65.3(1.3) & 61.7(2.2) \\
& \textbf{94.1(0.6)} & \textbf{87.6(0.6)} & 74.0(2.0) & 86.6(0.7) & 74.7(1.9) & 67.8(1.4) & 73.4(1.6)& 68.0(1.6) & 74.5(1.7) & 67.5(1.3) & 74.2(1.6) & 67.5(1.4) & 73.4(1.4) &  67.4(1.1) \\
& \textbf{94.8(0.6)} & 88.7(0.5) & 88.5(2.0) & \textbf{91.0(0.7)} & 88.3(1.1) & 82.3(1.3) & 88.9(0.7)& 79.9(1.0) & 88.8(0.8) & 80.1(0.8) & 89.4(0.9) & 81.4(0.7) & 88.1(0.8) &  82.3(1.2) \\
& 95.7(0.3) & 89.4(0.8) & 95.6(0.6) & \textbf{92.4(0.5)} & 95.2(0.5) & 90.5(0.8) & 95.4(0.4)& 90.0(0.9) & 95.6(0.5) & 89.0(0.4) & \textbf{96.2(0.4)} & 90.2(0.5) & 95.1(0.6) & 90.5(0.7) \\
& 95.4(0.3) & 89.8(0.5) & 97.4(0.3)& 93.0(0.5) & 97.2(0.3) & 92.9(0.4) & 97.3(0.2)& \textbf{93.6(0.3)} & 97.4(0.3) & 92.3(0.3) & \textbf{97.7(0.2)} & 93.3(0.4) & 97.4(0.4) &  92.6(0.6) \\
\hline
\multirow{7}{*}{dog} & \textbf{58.9(3.0)} &\textbf{65.0(7.5)} & 52.5(1.8) & 59.4(5.3) & 51.2(1.1) & 52.1(1.5) & 51.8(1.5)& 52.2(1.9) & 51.7(1.4) & 51.4(1.3) & 51.2(1.3) & 52.2(1.2) & 51.9(1.9) &  52.3(1.0) \\
& \textbf{89.5(2.0)} & \textbf{84.8(2.1)} & 59.6(2.0) & 79.4(1.9) & 58.6(1.8) & 56.8(1.8) &60.0(1.0)& 56.2(1.5) & 59.6(1.8) & 56.1(0.7) & 59.5(1.6) & 56.1(1.5) & 59.1(1.2) &  55.6(1.9) \\
& \textbf{94.2(0.7)} & \textbf{88.1(0.7)} & 69.7(3.8) & 84.8(1.0) & 67.3(1.9) & 63.0(1.4) & 67.7(0.8)& 62.0(1.0) & 68.3(1.3) & 61.8(0.8) & 67.2(2.3) & 62.4(1.0) & 68.5(1.0) &  62.6(1.5) \\
& \textbf{95.6(0.7)} & \textbf{89.2(1.6)} & 76.3(2.4) & 88.6(0.7) & 74.0(1.1) & 68.3(1.5) & 74.6(1.0)& 68.0(1.6) & 74.6(0.5) & 66.7(1.6) & 74.7(1.1) & 68.3(0.9) & 74.4(1.4) &  68.2(1.5) \\
& \textbf{96.0(0.4)} & 90.9(1.3) & 89.3(1.3) & \textbf{91.9(0.8)} & 88.9(0.9) & 83.9(0.8) & 89.4(0.7)& 81.0(0.8) & 89.6(0.8) & 81.4(1.4) & 90.5(0.7) & 82.4(1.1) & 89.01(0.8) &  83.5(0.9) \\
& 96.4(0.4) & 91.4(0.5) & 95.6(0.5)& \textbf{93.3(0.3)} & 95.5(0.4) & 91.7(0.5) & 95.7(0.3)& 90.9(0.7) & 95.6(0.4) & 89.7(0.7) & \textbf{96.9(0.3)} & 91.4(0.6) & 95.5(0.4) &  91.5(0.6)\\
& 96.5(0.3) & 91.8(0.2) & 97.6(0.1) & 93.7(0.4) & 97.4(0.2) & 93.2(0.4) & 97.5(0.2)& 93.7(0.5) & 97.5(0.2)& 92.3(0.7) & \textbf{98.0(0.1)} & \textbf{94.0(0.3)} & 97.3(0.3) & 93.1(0.7)\\
\hline
\multirow{7}{*}{frog} & \textbf{53.6(4.0)} & \textbf{67.3(5.8)} & 51.4(1.8) & 62.7(6.9) & 51.8(1.9) & 51.6(1.6) & 51.4(2.2)& 52.3(1.8) & 51.3(1.4) & 51.9(1.6) & 51.0(1.9) & 52.1(1.7) & 51.2(2.3) & 51.5(1.3) \\
& \textbf{92.5(1.8)} & \textbf{88.0(1.7)} & 59.8(2.8) & 82.4(2.0) & 59.5(1.8) & 56.6(1.5) & 59.2(1.8)& 56.9(1.5) & 59.2(1.6) & 56.9(0.9) & 60.8(2.2) & 56.7(0.9) & 59.7(2.6) &  57.4(2.0) \\
& \textbf{96.2(0.4)} & \textbf{91.1(0.8)} & 69.4(2.7) & 89.0(0.9) & 67.3(1.6) & 63.0(1.6) & 68.6(1.1)& 62.1(1.0) & 68.6(2.3) & 63.6(1.7) & 67.9(1.3) & 63.2(1.5) & 68.2(1.6) & 63.0(1.7) \\
& \textbf{96.5(0.6)} & \textbf{92.4(0.3)} & 76.8(2.7) & 91.3(0.8) & 75.5(2.3) & 69.3(1.2) & 75.7(1.5)& 69.4(1.6) & 75.4(1.6) & 68.6(1.6) & 75.4(1.4) & 69.5(1.1) & 75.8(1.3) &  69.1(1.3) \\
& \textbf{97.0(0.3)} & 93.0(0.3) & 90.5(1.4) & \textbf{93.8(0.4)} & 91.0(0.5) & 85.2(1.4) & 91.3(0.8)& 82.9(1.1) & 90.9(0.7) & 82.7(0.8) & 91.8(0.8) &84.6(0.4) & 90.7(0.8) &  85.6(0.9) \\
& 97.2(0.4) & 93.3(0.3) & 96.9(0.4) & \textbf{95.2(0.3)} & 96.5(0.4) & 93.2(0.6) & 96.8(0.3)& 92.8(0.6) & 96.6(0.3) &91.6(0.5) & \textbf{97.8(0.1)} & 93.2(0.4) & 96.6(0.2) & 93.5(0.5) \\
& 97.6(0.2) & 93.6(0.4) & 98.5(0.3) & \textbf{95.6(0.3)} & 98.3(0.2) & 95.0(0.3)& 98.3(0.2)& 95.4(0.5) & 98.4(0.2) & 94.4(0.5) & \textbf{98.8(0.1)} & \textbf{95.6(0.4)} & 98.4(0.2) &  95.2(0.3) \\
\hline
\multirow{7}{*}{horse} & \textbf{56.1(3.9)} & \textbf{60.1(4.5)} & 52.2(1.8) & 55.3(2.6) & 52.0(1.7) & 51.7(1.0) & 50.9(1.4)& 50.8(1.1) & 51.3(1.8) & 51.5(1.2) & 51.6(1.3) & 51.2(1.1) & 51.4(1.7) &  50.9(1.4) \\
& \textbf{87.7(2.7)} & \textbf{83.3(2.0)} & 58.8(2.0) & 77.6(1.9) & 59.0(1.4) &55.6(1.2) & 59.1(2.5)& 56.3(1.3) & 59.2(2.9) & 56.1(2.0) & 59.6(2.4) & 55.6(1.0) & 58.8(1.9) &  56.0(1.4) \\
& \textbf{94.4(0.5)} & \textbf{87.3(0.7)} & 68.1(2.5) & 83.3(1.1) & 68.5(1.5) & 62.2(1.6) & 68.4(1.4)& 61.8(1.2) & 68.4(1.5) & 63.3(1.6) & 68.9(1.1) & 62.5(1.4)& 67.6(1.4) &  61.8(1.5) \\
& \textbf{96.1(0.5)} & \textbf{89.7(1.1)} & 74.9(2.5) & 88.3(0.7) & 74.2(2.0) & 69.2(1.3) & 75.0(1.9)& 68.0(1.1) & 75.0(1.9) & 69.2(1.1) & 74.6(1.8) & 68.9(1.3) & 74.8(1.5) &  68.9(1.8) \\
& \textbf{96.8(0.3)} & 91.6(0.7) & 90.1(0.6) & \textbf{92.7(0.4)} & 89.7(1.0) & 84.0(1.3) & 89.8(0.6)& 81.0(1.1) & 89.5(0.8) & 81.6(1.1) & 91.0(0.9) & 82.3(0.9) & 89.2(1.0) & 83.8(1.3) \\
& \textbf{97.2(0.2)} & 91.9(1.3) & 96.1(0.5) & \textbf{94.4(0.2)} & 96.0(0.4) & 91.9(0.5) & 96.1(0.4)& 91.7(0.9) & 96.1(0.3) & 90.7(0.6) & 96.9(0.4) & 91.7(0.5) & 95.8(0.5) & 91.4(0.8) \\
& 97.3(0.2) & 92.7(0.3) & 97.9(0.2)& \textbf{94.5(0.2)} & 97.8(0.2) & 93.7(0.6) & 98.0(0.2)& 94.2(0.4) & 98.0(0.2) & 93.3(0.4) & \textbf{98.4(0.1)} & \textbf{94.5(0.4)} & 97.8(0.2) & 93.4(0.5) \\
\hline
\multirow{7}{*}{ship} & \textbf{54.3(2.3)} & \textbf{54.5(4.6)} & 52.5(2.2) & 54.6(4.0) & 50.8(1.4) & 51.8(1.4) & 51.8(1.8)& 51.5(1.5) & 51.6(1.9) & 51.2(1.3) & 51.9(1.8) & 51.5(0.8) & 51.5(1.8) &  51.3(1.6) \\
& \textbf{70.3(2.5)} & \textbf{67.4(3.3)} & 58.2(2.4) & 66.5(1.9) & 57.2(0.9) & 55.3(1.5) & 56.8(1.4)& 55.6(1.0) & 57.2(1.6) & 55.1(1.0) & 56.5(1.6) & 54.9(1.2) & 56.6(1.8) & 55.3(1.1) \\
& \textbf{82.6(1.5)} & \textbf{73.9(1.3)} & 66.4(2.2) & 72.6(1.1) & 64.6(2.1) & 60.1(1.2) & 64.0(1.4)& 60.7(1.4) & 64.0(1.2)& 59.7(1.1) & 64.2(1.8) & 60.4(0.8) & 64.5(1.4) &  59.9(1.2) \\
& \textbf{87.5(1.1)} & 76.2(0.5) & 74.3(1.3) & \textbf{77.3(1.3)} & 72.1(1.0) & 65.4(0.8) & 71.4(1.4)&65.0(0.9) & 71.5(1.2) & 65.2(1.1) & 71.2(0.8) & 65.3(1.6) & 72.1(1.3) & 64.9(1.2) \\
& \textbf{91.1(0.6)} & 79.9(0.6) & 85.5(1.1) & \textbf{83.7(0.6)} & 83.6(1.3) & 75.9(0.5) &  83.9(1.3)& 76.1(1.2) & 84.1(1.1) & 75.3(1.1) & 83.7(0.9) & 75.5(1.0) & 83.6(1.4) &  75.9(1.1) \\
& \textbf{92.8(0.3)} & 81.3(0.9) & 92.3(0.7) & \textbf{86.5(0.5)} & 91.7(0.5) & 84.3(0.8) & 92.0(0.5)& 85.2(0.7) & 91.7(0.6) & 83.7(0.7) & 92.3(0.7) & 84.6(0.9) & 91.9(0.4) & 84.3(0.7) \\
& 92.4(0.4) & 82.0(1.0) & \textbf{94.5(0.4)} & 86.9(0.5) & 94.0(0.4) & 87.7(0.6) & 94.2(0.4)& 87.8(0.8) & 94.3(0.5) & 87.6(0.5) & 94.4(0.4) & \textbf{88.0(0.5)} & 94.1(0.6) & 87.3(0.6) \\
\hline
\multirow{7}{*}{truck} & \textbf{55.8(5.6)} & \textbf{58.9(4.3)} & 51.8(1.9) & 58.7(2.7) & 51.6(1.7) & 51.4(1.7) & 51.3(1.8)& 51.7(1.9) & 52.4(1.6) & 51.3(1.2) & 51.8(1.9) & 51.5(1.6) & 51.5(2.0) &  51.3(1.5) \\
& \textbf{80.6(4.4)} & \textbf{77.9(1.4)} & 59.1(1.6) & 71.6(1.1) & 58.2(2.2) & 55.5(1.7) & 58.0(1.6)& 55.8(1.6) & 59.4(2.1) & 55.5(1.0) & 58.4(1.7) & 55.3(1.1) & 58.3(1.3) &  56.3(1.5)\\
& \textbf{91.5(0.8)} & \textbf{83.5(1.2)} & 66.4(2.2) & 80.5(1.3) & 66.1(1.8) & 61.6(1.6) & 65.8(1.2)& 61.3(1.4) & 66.0(1.5) & 61.7(1.5) & 65.8(1.7) & 62.2(1.0) & 65.1(1.4) & 61.8(1.6) \\
& \textbf{93.0(0.5)} & \textbf{84.7(1.5)} & 73.4(1.5) & 84.6(1.2) & 73.4(1.5) & 67.3(1.2) & 72.9(1.2)& 67.0(1.2) & 73.1(1.8) & 67.0(1.4) & 73.5(1.5) & 67.3(1.1) & 72.3(1.1) & 67.5(1.4) \\
& \textbf{94.0(0.3)} & 87.0(0.8) & 87.5(0.8) & \textbf{88.4(0.6)} & 86.4(1.1) & 80.5(1.4) & 87.1(1.0)& 79.4(1.0) & 87.0(0.6) & 78.2(0.8) & 88.0(1.0) & 80.1(1.1) & 86.5(1.0) &  79.7(1.0) \\
& 94.3(0.5) & 87.4(0.6) & 94.2(0.7) & \textbf{88.9(2.8)} & 93.6(0.7) & 88.0(0.8) & 93.9(0.7)& 87.9(0.8) & 93.7(0.8) & 86.9(0.9) & \textbf{94.6(0.3)} & 88.3(0.6) & 93.7(0.6) &  88.0(0.8) \\
& 94.7(0.5) & 87.7(0.4) & 96.1(0.3) & 90.5(0.4) & 96.0(0.3) & 90.1(0.7) & 95.9(0.3)& 90.2(0.8) & 96.1(0.3) & 89.8(0.7) &\textbf{96.3(0.3)} & \textbf{90.8(0.4)} & 95.6(0.3) &  89.9(0.5) \\
\hline
\end{tabular}}
\end{table*}

\begin{table*}[htbp]
\label{table12}
\caption{mean and standard deviation of AUC/BER, where $n=1000$, $\epsilon=0.1,0.5,1,1.5,3,5,7$}
\centering
\setlength{\tabcolsep}{0.75mm}{
\begin{tabular}{|c|c|c|c|c|c|c|c|c|c|c|c|c|c|c|}
\hline
\multirow{2}{*}{Dataset} & \multicolumn{2}{c|}{Barrier} & \multicolumn{2}{c|}{Sigmoid} & \multicolumn{2}{c|}{Unhinged} & \multicolumn{2}{c|}{Savage} & \multicolumn{2}{c|}{Logistic} & \multicolumn{2}{c|}{Squared} & \multicolumn{2}{c|}{Hinge}  \\
&AUC& BER& AUC & BER & AUC & BER & AUC &BER & AUC &BER & AUC & BER & AUC & BER\\
\hline
\multirow{7}{*}{automobile} & \textbf{52.4(13.1)} & 48.6(10.0) & 49.4(5.2) & 50.7(6.2)& 50.5(3.7) & 51.7(3.9) & 52.0(6.2)& 50.3(4.2) & 48.8(6.6) & 52.7(3.2) & 50.8(4.2) & \textbf{52.8(4.9)} & 50.6(5.5) &  52.1(5.0) \\
& \textbf{65.0(5.9)} & \textbf{67.5(4.6)} & 55.5(6.1)& 61.1(5.8) & 56.8(3.8) & 55.0(4.0) & 55.4(6.0)& 56.2(3.1) & 55.2(7.1) & 54.1(3.8) & 57.6(5.7) & 55.8(4.4) & 56.2(4.2) &  54.0(2.6) \\
& \textbf{79.9(4.1)} & \textbf{75.8(2.5)} & 68.9(5.3) & 72.1(3.2) & 63.9(5.4) & 61.7(3.4) & 65.6(5.4)& 60.0(4.6) & 65.6(5.4) & 61.1(3.4) & 64.5(4.1) &59.5(5.6) & 63.5(5.7) &  61.8(3.5) \\
& \textbf{85.1(2.1)} & 75.2(1.4) & 72.4(2.7) & \textbf{75.4(2.2)} & 70.7(4.6) & 68.2(3.6) & 69.8(2.7)& 67.5(2.2) & 72.0(5.3) & 66.9(3.1) & 68.0(4.9) & 66.4(5.4) & 70.3(4.6) & 67.1(3.6) \\
& \textbf{88.0(1.2)} & 75.0(1.4) & 86.1(2.0) & \textbf{82.0(1.8)} & 85.4(3.1) & 78.2(3.0) & 84.0(3.7)& 76.5(2.7) & 84.5(2.3) & 76.6(2.2) & 83.2(3.9) & 76.1(2.6) & 84.5(2.7) &  77.8(2.6)\\
& 88.8(1.6) & 75.9(0.6) & 91.4(1.6) & 84.5(1.0) & 91.6(0.6) & 84.3(2.0) & 91.2(1.7)& 83.7(1.3) & 91.3(1.6) & 84.0(2.7) & 90.5(1.5) & 84.0(2.6) & \textbf{91.9(1.7)} &  \textbf{84.9(2.1)} \\
& 90.2(1.0) & 75.6(0.9) & 93.2(0.7) & 85.2(1.1) & 94.2(1.3) & 86.4(1.7) & \textbf{94.5(1.2)}& \textbf{86.7(1.4)} & 94.3(1.0) & 86.1(2.0) & 94.2(1.2) & 86.4(1.2) & 94.1(1.6) & 85.6(1.3) \\
\hline
\multirow{7}{*}{bird} & 53.7(7.9) & 49.5(9.5) & 53.5(4.1) & \textbf{49.7(6.6)} & 52.2(6.7) & 49.4(5.1) & \textbf{53.8(6.0)}& 49.5(3.4) & 53.2(5.4) & 48.7(5.9) & 52.9(5.8) & 49.4(5.1) & 51.4(4.2) &  48.6(5.1) \\
& \textbf{64.9(5.4)} & \textbf{60.5(6.4)} & 56.6(5.3) & 60.0(3.2) & 58.0(4.1) & 55.7(3.8) & 58.0(5.4)& 56.7(4.0) & 54.7(5.5) & 56.1(4.3) & 56.9(4.8) & 54.8(3.0) & 56.4(5.9) & 56.4(3.1) \\
& \textbf{77.2(2.8)} & \textbf{69.8(4.1)} & 65.3(4.8) & 68.7(3.4) & 63.9(5.1) & 59.6(3.5) &  63.4(4.3)& 60.4(2.0) & 64.6(5.0) & 61.0(3.5) & 62.4(5.3) & 58.7(3.4) & 63.3(4.5)&  60.8(3.5) \\
& \textbf{78.7(2.4)} & 71.0(1.6) & 71.0(5.6) & \textbf{72.8(3.3)} & 68.5(4.7) & 64.3(3.3) &  67.3(5.4)& 63.5(2.6) & 69.0(3.7) & 64.6(2.6) & 70.4(3.2) & 64.6(3.5) & 68.6(5.5) & 64.7(3.7) \\
& 81.3(1.1) & 70.5(1.7) & \textbf{83.5(1.6)} & 76.2(1.6) & 79.6(3.8) & \textbf{76.3(4.1)} & 79.6(3.9)& 76.0(3.2) & 81.4(2.7) & 74.0(3.2) & 79.7(3.1) & 73.9(3.9) & 79.9(3.7) &  74.2(3.4) \\
& 81.0(0.7) & 73.6(1.7) & 87.6(1.0) & 77.3(1.3) & 87.5(2.0) & 79.9(1.5) & 87.4(1.7)& 78.5(1.9) & 87.7(1.7) & 79.9(2.0) & 87.7(1.9) & \textbf{80.0(1.2)} & \textbf{88.2(1.8)} &  79.5(1.8) \\
& 81.7(0.9) & 73.3(1.1) & \textbf{89.7(0.6)} & 78.6(1.7) & 88.7(2.1) & 81.3(1.4) & 89.4(1.5)& 80.5(2.5) & 89.4(1.1) & 81.3(1.8) & 89.0(1.1) & \textbf{81.6(1.1)} & 88.7(1.6) & 81.3(1.5) \\
\hline
\multirow{7}{*}{cat} & 52.5(7.7) & \textbf{51.5(13.2)} & 52.6(7.7) & 48.7(6.6) & 53.5(5.6) & 49.2(5.8) & \textbf{53.6(8.1)}& 51.1(3.8) & 51.0(6.7) & 50.5(1.9) & 52.5(6.0) & 50.7(4.3) & 52.4(5.9) & 51.0(4.0) \\
& \textbf{76.0(7.1)} & \textbf{73.6(4.5)} & 61.1(3.8) & 65.3(8.2) & 58.7(5.1) & 56.9(4.0) & 60.0(4.8)& 56.5(3.4) & 61.4(5.6) & 55.3(6.5) & 60.4(3.8) & 55.3(2.7) & 59.4(2.6) & 58.1(4.0) \\
& \textbf{84.4(2.4)} & \textbf{77.9(2.5)} & 67.7(3.2) & 76.1(3.3) & 65.8(3.4) & 60.7(4.3) & 65.3(4.9)& 61.6(4.6) & 66.2(3.8) & 61.7(2.3) & 63.5(3.9) & 60.5(4.5) & 67.6(1.7) &  60.9(6.0) \\
& \textbf{87.0(2.0)} & \textbf{80.2(1.8)} & 74.5(4.0) & 77.9(3.8) & 71.0(4.1) & 66.2(3.5) & 72.7(3.6)& 66.5(3.1) & 73.4(3.7) & 66.3(2.2) & 70.9(4.2) & 65.9(3.6)& 71.3(4.1) &  67.5(2.2) \\
& \textbf{87.8(1.4)} & 81.4(1.7) & 85.9(2.0) & \textbf{83.1(1.6)} & 84.1(3.0) & 80.0(2.2) & 84.7(2.8)& 76.5(2.5) & 85.3(2.3) & 77.0(3.0) & 84.4(3.3) & 76.7(1.8) & 84.6(3.1) &  78.8(2.9) \\
& 88.8(0.8) & 81.3(1.1) & 89.3(1.5) & \textbf{83.9(1.4)} & 89.8(1.7) & 82.9(2.5) & 89.7(1.5)& 83.4(2.4) & 90.1(1.5) & 82.3(2.6) & \textbf{90.4(1.6)} & \textbf{83.9(1.3)} & 89.8(2.2) &  83.6(1.3) \\
& 88.3(1.2) & 82.8(0.8) & 90.9(1.3) & 84.0(1.5) & 90.9(1.0) & 85.2(1.1) & 90.9(1.2)& 85.1(1.6) & 91.2(1.4)& 84.7(1.2) & \textbf{91.7(1.3)} & \textbf{85.8(0.9)} & 91.0(1.2) &  84.3(0.8) \\
\hline
\multirow{7}{*}{deer} & \textbf{63.7(9.2)} & \textbf{58.0(10.0)} & 53.1(4.3) & 56.0(6.1)& 52.2(5.1) & 51.1(5.4) & 52.0(5.1)& 52.7(3.8) & 52.4(6.0) & 51.4(4.8) & 50.6(2.6) & 52.3(4.5) & 52.1(5.7) & 51.9(5.5) \\
& \textbf{75.2(8.2)} & \textbf{69.3(8.9)} & 60.2(4.8) & 64.3(6.6) & 60.6(4.0) & 51.7(4.0) & 60.3(4.0)& 53.1(4.0) & 60.8(3.0) & 53.3(4.5) & 61.3(3.4) & 53.8(3.9) & 60.6(5.6) &  54.2(4.5) \\
& \textbf{83.5(2.1)} & \textbf{76.3(3.1)} & 63.5(4.9) & 70.2(5.1) & 60.7(6.6) & 57.3(1.8) & 60.3(4.7)& 57.8(4.3) & 61.2(5.8) & 58.1(3.8) & 60.5(1.8) & 58.8(4.2) & 62.6(4.3) & 58.3(3.9) \\
& \textbf{85.9(1.4)} & \textbf{79.1(1.7)} & 70.1(4.0) & 78.2(1.6) & 68.6(3.9) & 65.2(2.8) & 67.7(4.4)& 62.7(2.7) & 69.3(5.3) & 63.6(3.0) & 70.2(3.2) & 62.6(5.4) & 67.9(4.5) & 63.9(3.3) \\
& \textbf{85.6(1.1)} & \textbf{80.3(1.0)} & 81.9(3.1) & 79.1(1.5) & 80.1(3.1) & 76.2(2.8) & 81.9(3.7)& 73.7(2.8) & 81.2(2.7) & 74.5(3.6) & 82.1(3.1) & 74.1(4.3) & 80.0(2.6) &  75.2(3.3) \\
& 85.8(0.8) & 80.1(1.1) & \textbf{87.5(1.0)} & 82.1(1.1) & 87.2(2.5) & 82.2(1.5) & 86.2(1.9)& 81.1(2.4) & 86.6(1.8) & 80.9(2.3) & 87.1(1.3) & \textbf{82.8(1.8)} & 86.6(2.8) & \textbf{82.8(1.9)} \\
& 86.6(1.7) & 80.9(1.0) & 89.6(1.0) & 82.4(1.6) & 89.8(1.6) & 82.9(1.4) & 90.2(1.6)& 82.7(1.6) & \textbf{91.2(1.4)} & 81.7(1.8) & 89.8(1.6) & \textbf{83.4(1.6)} & 89.9(1.9) &  \textbf{83.4(2.0)}\\
\hline
\multirow{7}{*}{dog} & \textbf{57.5(12.2)} & \textbf{56.5(7.3)} & 50.4(6.0) & 53.0(4.9) & 50.4(5.4) & 49.0(4.9) & 49.9(3.9)& 49.7(4.1) & 52.9(5.7) & 50.5(3.4) & 51.6(3.8) & 51.0(4.6) & 52.9(6.5)&  50.1(5.0) \\
& \textbf{71.8(7.0)} & \textbf{72.5(3.0)} & 57.0(4.3) & 64.7(4.0) & 58.7(6.2) & 56.7(4.9) & 59.1(5.4)& 55.2(3.1) & 56.6(5.5) & 58.4(4.2) & 55.2(6.4) & 55.5(3.3) & 56.7(5.4) & 58.4(4.2) \\
& \textbf{83.8(3.8)} & \textbf{77.3(2.6)} & 68.0(5.0) & 71.4(4.6) & 64.8(4.5) & 62.0(2.9) & 66.6(4.6)& 60.3(2.5) & 65.3(6.0) & 60.8(6.1) & 65.6(5.4)& 59.6(3.4) & 63.9(7.2) & 59.9(3.6) \\
& \textbf{87.2(1.6)} & 80.6(2.0) & 74.6(3.6) & \textbf{81.2(2.0)} & 70.9(4.6) & 68.8(2.8) & 72.3(5.0)& 68.0(2.2) & 72.7(5.4) & 66.9(2.1) & 72.7(4.5) & 67.4(3.5) & 72.7(4.9) & 69.7(1.8) \\
& \textbf{88.9(1.4)} & 81.6(1.8) & 86.0(2.9) & \textbf{83.7(1.6)} & 83.0(2.2) & 80.6(3.9) & 84.0(2.7)& 79.2(3.0) & 85.3(2.7) & 79.2(2.6) & 82.8(2.2) & 78.9(3.1) & 84.3(3.1) &  80.1(2.9) \\
& 90.1(1.0) & 82.9(1.2) & \textbf{92.8(1.5)} & 86.6(1.2) & 92.1(1.6) & \textbf{87.6(1.3)} & 92.7(2.0)& 86.1(2.2) & 92.3(1.4) &85.8(2.8) & 92.5(1.1) & 86.8(2.0) & \textbf{92.8(1.2)} & 87.3(2.0) \\
& 89.0(1.0) & 83.0(0.9) & 94.5(1.0) & 86.8(1.3) & 93.8(1.4) & 87.5(1.9) & 94.2(1.0)& \textbf{88.6(1.8)} & 94.3(1.2) & 87.0(2.4) & \textbf{95.0(0.9)} & 88.2(1.7) & 94.3(1.2) &  88.1(1.7)\\
\hline
\multirow{7}{*}{frog} & 49.0(18.7) & \textbf{58.3(14.3)} & \textbf{50.7(7.6)} & 56.5(5.6) & 49.6(8.0) & 54.7(4.0) & 48.6(5.8)& 52.5(4.3) & 47.6(6.3) & 52.4(4.5) & 47.7(7.3) & 51.7(2.7) & 48.9(7.0) & 54.7(4.2) \\
& \textbf{79.9(10.2)} & \textbf{80.7(5.5)} & 61.0(9.2) & 76.2(7.4) & 57.7(6.5) & 59.8(3.9) & 58.0(5.7)& 57.0(6.0) & 59.2(5.7) & 57.3(5.8) & 58.0(6.4) & 57.8(4.2) & 57.6(6.8) &  58.6(4.7) \\
& \textbf{91.1(1.3)} & \textbf{85.3(1.3)} & 71.8(4.1) & 83.2(2.1) & 69.1(4.7) & 66.5(3.5) & 67.6(4.7)& 66.0(5.3) & 69.9(4.5) & 67.5(4.0) & 67.3(4.2) & 66.3(4.2) & 68.8(5.3) & 66.3(2.7) \\
& \textbf{92.6(0.6)} & \textbf{86.4(1.6)} & 77.0(6.5) & 85.7(1.5) & 72.7(5.8) & 70.4(2.4) & 74.2(3.5)& 67.7(3.6) & 75.1(4.9) & 70.6(3.6) & 72.3(5.3) & 69.5(2.0) & 73.2(6.2) &  70.6(5.3) \\
& \textbf{93.1(1.0)} & 87.0(1.3) & 91.0(2.2)& \textbf{88.9(1.0)} & 86.6(2.8) & 84.9(1.9) &  87.1(2.6)& 81.2(3.7) & 89.4(3.0) & 82.5(2.9) & 88.7(2.2) & 82.2(2.8) & 87.4(2.3) & 85.5(2.3) \\
& 93.8(0.6) & 87.8(1.1) & 94.3(0.7) & \textbf{90.2(1.1)} & 92.7(1.6) & 89.2(2.0) & 94.0(0.8)& 88.4(1.6) & \textbf{94.6(0.7)} & 88.6(1.4) & 94.3(1.4) & 88.7(1.6) & 94.1(1.5) &  90.0(1.4) \\
& 93.6(0.5) & 88.2(0.7) & 95.1(0.9) & 89.2(1.0) & 95.2(1.3) & 90.5(1.0) & 95.4(1.4)& \textbf{90.6(1.4)} & 95.8(0.6) & 89.6(1.3) & \textbf{95.9(0.9)} & 90.4(1.5) & 95.1(0.8) &  89.8(1.0) \\
\hline
\multirow{7}{*}{horse} & \textbf{55.6(11.1)} & \textbf{53.1(10.1)} & 52.6(2.9) & 51.6(6.8) & 52.2(5.3) & 53.1(5.2) &  51.7(3.3)& 51.9(2.8) & 50.1(5.3)& 50.3(4.4) & 53.7(3.2) & 50.1(3.7) & 51.7(4.6) &  52.6(4.9) \\
& \textbf{73.4(6.5)} & \textbf{70.4(8.2)} & 60.4(3.5) & 64.8(4.0) & 60.3(4.0) & 56.0(3.5) & 59.4(2.9)& 57.2(3.6) & 60.2(3.2)& 58.5(4.4) & 58.7(2.7) & 56.8(3.3) & 59.8(2.5) & 58.0(3.5) \\
& \textbf{84.0(3.5)} & \textbf{77.7(4.2)} & 68.2(7.5) & 75.4(4.6) & 62.5(6.0) & 63.2(4.7) & 63.4(5.0)& 63.6(3.8) & 64.9(5.0) & 63.8(4.9) & 64.8(5.2) & 62.3(3.8) & 64.8(5.8) &  64.9(4.0) \\
& \textbf{87.3(2.2)} & \textbf{80.2(1.6)} & 75.2(4.7)& 78.7(2.7) & 71.8(4.8) & 67.3(2.2) & 70.8(4.9)&65.4(1.4)& 73.0(5.5) & 64.8(3.0) & 70.3(3.4) & 65.4(2.6) & 70.75(5.8) & 67.8(3.5) \\
& \textbf{89.3(1.0)} & 81.4(1.6) & 87.4(2.1) & \textbf{83.2(1.9)}  & 83.7(2.9) & 81.1(2.7) & 84.7(3.3)& 79.4(2.4) & 84.7(3.3) & 79.7(1.8) & 84.7(3.1) & 78.4(3.0) & 84.9(1.7) & 80.3(2.6)\\
& 90.3(0.8) &82.1(1.1) & \textbf{93.5(1.2)} & 85.4(1.5) & 91.9(1.2) & 86.5(2.1) & 92.5(1.3)& 86.0(1.9) & 92.5(1.0) & 85.4(1.9) & 92.9(1.7) & 85.5(1.8) & 92.8(1.3) &  \textbf{86.8(1.4)}\\
& 90.9(0.7) &83.0(1.9) & \textbf{94.7(1.3)} &84.8(1.2) & 94.4(1.4) & \textbf{88.0(1.8)} & 94.2(0.9)& 87.7(1.5) & \textbf{94.7(1.1)} & 87.7(1.3) & 94.3(1.2) & 87.9(0.9) & 94.0(1.2) & \textbf{88.0(2.1)} \\
\hline
\multirow{7}{*}{ship} & 49.2(7.7) & 49.3(6.0) & 49.1(5.5) & 49.1(5.7) & \textbf{49.9(4.6)} & \textbf{51.7(3.5)} & 47.5(6.1)& 51.4(3.1) & 49.1(5.1) & \textbf{51.7(3.6)} & 48.7(3.3) & 49.1(3.0)& 48.0(4.6) & 50.6(3.0)\\
& \textbf{60.1(5.2)} & \textbf{57.5(6.2)} & 57.5(5.8) & 55.9(5.0) & 58.4(5.6) & 53.5(3.2) & 55.8(5.4)& 54.6(4.5) & 56.9(4.7)& 54.0(2.7) & 57.0(4.1) & 53.9(3.8) & 56.8(3.9) & 53.6(3.5) \\
& \textbf{66.7(4.7)} & 62.0(2.7) & 59.7(4.6) &\textbf{ 62.3(4.7)} & 56.8(5.5) & 57.3(4.8) & 59.6(4.5)& 55.6(4.4) & 57.7(7.2) & 56.5(2.6) & 58.1(5.1) & 56.8(4.3) & 59.2(3.9) & 56.4(2.5) \\
& \textbf{67.5(4.4)} & 65.1(2.6) & 62.9(4.8) & \textbf{65.5(3.3)} & 61.2(4.0) & 63.6(4.6) & 61.7(3.9)& 60.6(3.8) & 61.5(4.6) & 61.7(3.2) & 63.1(4.0) & 63.1(4.3) & 61.5(4.9) & 61.9(4.4) \\
& 70.8(2.0) & 66.0(1.7) & \textbf{74.7(2.4)}& 70.2(2.4) & 73.1(3.4) & 68.1(2.7) & 72.9(2.1)& 68.1(2.4) & 73.2(2.9) & \textbf{70.8(3.3)} & 72.6(4.2) &68.1(2.7) & 72.9(2.8) & 67.8(3.9) \\
& 72.7(1.7) & 67.9(1.0) & 80.5(3.3) & 72.7(1.1) & 79.7(2.2) & 73.8(2.4) & 79.4(1.9)& 73.8(2.2) & 79.7(2.5) & 73.5(1.8) & \textbf{81.0(2.7)} & 74.4(2.3) & 78.3(2.6) & \textbf{75.5(2.4)} \\
& 73.9(1.5) & 67.3(1.3) & \textbf{83.2(1.9)} & 73.9(1.5) & 82.6(1.8)& 75.5(3.0) & 82.8(1.6)& 75.5(1.6) & 82.8(1.3) & 75.8(2.2) & 82.4(1.5) & 75.3(2.1) & 82.8(1.5) & \textbf{76.8(2.3)} \\
\hline
\multirow{7}{*}{truck} & \textbf{53.1(7.8)} & 49.3(8.3) & 51.7(4.6) & 49.3(3.9) & 51.8(5.2) & 50.9(4.2) & 51.1(3.7)& 50.0(2.4) & 52.6(5.5) & \textbf{51.4(4.1)} & 51.4(4.0) & 51.3(5.7) & 50.7(4.2) & 50.7(4.7) \\
& \textbf{63.6(9.2)} & \textbf{66.5(6.8)} & 56.3(7.0) & 60.2(6.4) & 56.4(5.1)& 54.9(4.2) & 54.8(5.5)& 57.1(4.9) & 55.4(7.0) & 57.7(3.0) & 57.2(7.1) & 54.4(3.2) & 55.0(5.7) & 56.8(3.5) \\
& \textbf{79.9(3.1)} & \textbf{72.6(2.3)} & 67.1(4.9) & 70.1(4.5) & 63.9(4.0) & 61.1(4.2) & 63.3(4.5)& 60.2(3.4) & 65.4(3.3) & 61.1(3.4) & 64.3(3.1) & 59.7(2.8) & 64.2(6.1) &  61.2(5.9)\\
& \textbf{83.6(1.8)} & \textbf{73.0(1.3)} & 75.5(3.0) & 72.7(2.7) & 71.8(4.2) & 64.9(5.0) & 72.7(3.6)& 66.0(4.7) & 73.4(4.0) & 66.0(5.5) & 70.7(2.2) & 65.4(4.0) & 71.5(3.9) & 66.5(4.4) \\
& \textbf{85.7(1.3)} & 73.7(0.6) & 84.4(3.4) & \textbf{79.8(2.6)} & 82.0(3.7) & 76.1(2.5) & 82.0(2.1)& 76.4(3.0) & 82.0(4.2) & 76.2(2.4) & 81.4(1.8) & 75.4(3.1) & 81.6(3.2) & 77.1(3.2) \\
& 86.9(1.8) & 75.2(1.5) & \textbf{89.5(1.1)} & \textbf{82.2(1.4)} & 89.0(1.9) & 80.5(1.2) & 88.9(1.4)& 81.2(1.9) & 89.1(2.0) & 79.9(1.8) & 88.2(1.4) & 80.2(2.0) & 88.1(1.0) &  81.4(2.4) \\
& 86.2(1.5) & 76.0(1.3) & 90.2(1.7) & 82.9(1.3) & 90.9(0.9) & \textbf{83.9(1.4)} &  90.5(1.5)& 83.6(0.7) & 90.7(0.8) & 82.5(1.5) & \textbf{91.2(0.8)} & 82.6(1.8)& \textbf{91.2(1.0)} &  83.7(0.9) \\
\hline
\end{tabular}}
\end{table*}

\end{appendix}

\end{document}